\documentclass[final]{dmtcs-episciences}

\usepackage{amsmath}
\usepackage{amsthm}
\usepackage{multirow}
\usepackage[disableredefinitions]{complexity}
\usepackage[utf8]{inputenc}
\usepackage{xcolor}
\usepackage[capitalise]{cleveref}
\usepackage{subcaption}
\usepackage{hyperref}
\usepackage{path}
\usepackage[round]{natbib}
\newcommand{\doi}[1]{doi: \href{https://doi.org/#1}{\nolinkurl{#1}}}

\makeatletter
\def\namedlabel#1#2{\begingroup
   \def\@currentlabel{#2}%
   \label{#1}\endgroup
}
\makeatother

\title{The Price of Upwardness\,\thanks{This work was initiated at the Dagstuhl Seminar ``Beyond-Planar Graphs: Models, Structures and Geometric Representations'' (No.\ 24062), February 2024.}}

\author[Patrizio Angelini et al.]{
  Patrizio Angelini\affiliationmark{1}
  \and Therese Biedl\affiliationmark{2}\thanks{Research was supported by NSERC FRN RGPIN-2020-03958.}\;
  \and Markus Chimani\affiliationmark{3}
  \and Sabine Cornelsen\affiliationmark{4}
  \and Giordano Da Lozzo\affiliationmark{5}\thanks{Supported, in part, by MUR of Italy (PRIN Project no.~2022ME9Z78 ~-- NextGRAAL and PRIN Project no.~2022TS4Y3N~-- EXPAND).}\;
  \and Seok-Hee Hong\affiliationmark{6}
  \and Giuseppe Liotta\affiliationmark{7}$^\ddagger$\thanks{Supported in part by MUR PON Proj. ARS01\_00540}\;
  \and Maurizio~Patrignani$^\ddagger$\affiliationmark{5}
  \and Sergey~Pupyrev\affiliationmark{8}
  \and Ignaz~Rutter\affiliationmark{9}\thanks{Supported by the Deutsche Forschungsgemeinschaft (DFG, German Research Foundation) -- 541433306.}\;
  \and Alexander~Wolff\affiliationmark{{10}}
}

\affiliation{
  ${}^{\phantom{0}1}$ John Cabot University, Rome, Italy\\
  ${}^{\phantom{0}2}$ University of Waterloo, Waterloo, Canada\\
  ${}^{\phantom{0}3}$ Universit\"at Osnabr\"uck, Osnabr\"uck, Germany\\
  ${}^{\phantom{0}4}$ University of Konstanz, Konstanz, Germany\\
  ${}^{\phantom{0}5}$ Roma Tre University, Rome, Italy \\
  ${}^{\phantom{0}6}$ University of Sydney, Sydney, Australia\\
  ${}^{\phantom{0}7}$ University of Perugia, Perugia, Italy\\
  ${}^{\phantom{0}8}$ Menlo Park, CA, U.S.A.\\
  ${}^{\phantom{0}9}$ University of Passau, Passau, Germany\\
  ${}^{10}$ Universität Würzburg, Würzburg, Germany
}

\keywords{Upward drawings, beyond planarity, upward $k$-planarity, upward outer-1-planarity}

\graphicspath{{figures/}}

\theoremstyle{plain}
\newtheorem{theorem}{Theorem} 

\newtheorem{lemma}{Lemma}
\newtheorem{proposition}{Proposition}
\newtheorem{corollary}{Corollary}
\newtheorem{observation}{Observation}

\definecolor{defblue}{rgb}{0.121,0.47,0.705}

\DeclareMathOperator{\bw}{bw}

\DeclareMathOperator{\lcr}{lcr^\uparrow}
\DeclareMathOperator{\skel}{skel}

\DeclareMathOperator{\refn}{refn}
\DeclareMathOperator{\expn}{expn}

\begin{document}

\publicationdata{vol.~27:3}{2025}{5}{10.46298/dmtcs.15222}{2025-02-11; 2025-07-10}{2025-07-17}

\maketitle

\begin{abstract}
        Not every directed acyclic graph (DAG) whose underlying undirected graph is planar admits an upward planar drawing. We are interested in pushing the notion of upward drawings beyond planarity
        by considering upward $k$-planar drawings of DAGs in which the edges are monotonically increasing in a common direction and every edge is crossed at most $k$ times for some integer $k \ge 1$. 
        We show that the number of crossings per edge in a monotone drawing is in general unbounded for the class of bipartite outerplanar, cubic, or bounded pathwidth DAGs. However, it is at most two for outerpaths and it is at most quadratic in the bandwidth in general. From the computational point of view, we prove that testing upward $k$-planarity is NP-complete already for $k =1$ and even for restricted instances for which upward planarity testing is polynomial. 
        On the positive side, we can decide in linear time whether a single-source DAG admits an upward $1$-planar drawing in which all vertices are incident to the outer face.
\end{abstract}

\renewcommand{\emph}[1]{{\color{defblue}\em #1}}

\section{Introduction}

Graph drawing ``beyond planarity'' studies the combinatorial and algorithmic questions related to representations of graphs where edges can cross but some crossing configurations are forbidden. Depending on the forbidden crossing configuration, different beyond-planar types of drawings can be defined including, for example,  RAC, $k$-planar, fan planar, and quasi planar drawings \citep*[see][for surveys and books]{DBLP:journals/csur/DidimoLM19,DBLP:books/sp/20/HT2020,DBLP:journals/csr/KobourovLM17}.

While most of the literature about beyond planar graph drawing has focused on undirected graphs (one of the few exceptions being the paper of \cite*{DBLP:journals/jgaa/AngeliniCDFBKS11}, which studies RAC upward drawings), we study \emph{upward $k$-planar drawings} of acyclic digraphs (DAGs), i.e.,  drawings of DAGs where the edges monotonically increase in $y$-direction 
and each edge can be crossed at most $k$ times.
The minimum $k$ such that a DAG admits an upward $k$-planar drawing is called its 
\emph{upward local crossing number}.
We focus on values of $k= 1,2$ and investigate both combinatorial properties and complexity questions. Our research is motivated by the observation that well-known DAGs that are not \emph{upward planar}, i.e., not upward $0$-planar, do admit a drawing where every edge is crossed at most a constant number of
times; see, e.g.,~\cref{fig:upward-planar-1-planar}.

\begin{figure}[htb]
        \centering
	\begin{subfigure}[b]{0.25\textwidth}
		\centering
		\includegraphics[page=2]{first-figure}
		\subcaption{}
		\label{fig:first-figure-a}
	\end{subfigure}
        \qquad
	\begin{subfigure}[b]{0.2\textwidth}
		\centering
		\includegraphics[page=1]{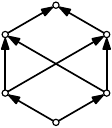}
		\subcaption{}
		\label{fig:first-figure-b}
	\end{subfigure} 

	\caption{A graph that is not upward planar but admits an upward 1-planar drawing.}
	\label{fig:upward-planar-1-planar}
\end{figure}

\paragraph*{Our contribution.}
\begin{itemize}
\item A graph is an \emph{outerpath} if it has a planar drawing in which each vertex is incident to the outer face and the internal
  faces induce a path in the dual graph.
  \cite{p-uptod-GD94} observed that there is a directed acyclic
8-vertex outerpath that is not upward-planar (see \cref{subfig:papakostas}). We strengthen this observation by showing that there exists a directed acyclic fan (that is, a very specific outerpath) that has no upward-planar drawing (\cref{prop:fan}).
On the other hand, we show that 
every directed acyclic outerpath is upward 2-planar (\cref{thm:outerpath}) and that the upward local crossing number is quadratic in the bandwidth (\cref{thm:bandwidth}). However, the upward local crossing number of bipartite outerplanar DAGs (\cref{thm:outer-unbounded}), bipartite DAGs with bounded pathwidth (\cref{cor:BipartitePathwidth}), and cubic DAGs (\cref{obs:cubic}) is in general unbounded. 
\item We show that upward $1$-planarity testing is NP-complete, even for graph families where upward planarity testing can be solved in polynomial time. These include: single-source single-sink series-parallel DAGs with a fixed rotation system; single-source two-sink series-parallel DAGs where the rotation system is not fixed; and  single-source single-sink DAGs without fixed rotation system that can be obtained from a $K_4$ by replacing the edges with series-parallel DAGs (\cref{th:st-negative}). 
\item Finally, following a common trend in the study of beyond planar graph representations,  we consider the \emph{outer model}, in which all vertices are required to lie on a common face while maintaining the original requirements~\citep*{DBLP:journals/csur/DidimoLM19,DBLP:books/sp/20/HT2020,DBLP:journals/csr/KobourovLM17}. 
We prove that testing whether a single-source DAG admits an upward outer-1-planar drawing can be done in linear time (\cref{th:outer-testing-short}).  
\end{itemize}

\paragraph*{Related Work.} 

A drawing of a graph is \emph{monotone} if all edges are drawn monotone with respect to some direction, e.g., a drawing is \emph{y-monotone} or \emph{upward}, if each edge intersects each horizontal line at most once.  
The corresponding crossing number is introduced and studied by~\cite{Valtr2007} and \citet*{FPSS2013}.
\cite{schaefer:17-22} mentions the upward crossing number and the local crossing number but not their combination.
\citet[p.~64]{schaefer:17-22} also showed that a drawing with the minimum number of crossings per edge can require incident edges that cross.  
The edges of the provided 4-planar example graph can be oriented such that the resulting directed graph admits an upward 4-planar drawing. 
Thus, also an upward drawing that achieves the minimum local crossing number can require incident edges that cross.
Also, the so-called strong Hanani--Tutte theorem carries over to directed graphs:  
\citet[Theorem~3.1]{FPSS2013} showed that every undirected graph that has a monotone drawing where any pair of
independent edges crosses an even number of times also has a planar monotone drawing with the same vertex positions.
This implies 
that in any upward drawing of a graph that is not upward-planar there must be a pair of independent edges that crosses an odd number of times.

Upward drawings of directed acyclic graphs have been studied in the context of (upward) book embeddings. In that
model the vertices are drawn on a vertical line (a spine) following a topological order of the graph,
while all edges are pointing upwards. To reduce the edge crossings, edges are partitioned into the fewest number of 
crossing-free subsets (pages). Studying upward book embeddings is a popular topic, which is usually centered
around determining the smallest number of pages for various graph classes \citep*{HeathPT99,DBLP:conf/gd/FratiFR11,FratiFR13,NollenburgP23,JungeblutMU23} 
or deciding whether a graph admits
an upward drawing with a given number of pages
\citep*{BinucciLGDMP19,BhoreLMN23,BekosLFGMR23}.
Our model is related to \emph{topological book embeddings}~\citep*{MS2009,GiordanoLMSW15}, which are a relaxed version of 2-page book embeddings in which edges are allowed to cross the spine. While papers about topological  book embeddings insist on planar drawings and minimize the number of spine crossings, we do allow crossings and want to bound the maximum
number of crossings per edge (ignoring the spine).

\section{Preliminaries}\label{sec:preli}

A \emph{drawing} $\Gamma$ of a graph~$G$ maps the vertices of~$G$ to
distinct points in the plane and the edges of $G$ to open Jordan curves
connecting their respective endpoints but not containing any other vertex point.
A \emph{crossing} between two edges is a common point of their curves, other than a common end point.
A drawing is \emph{simple} if (a) there are no self-crossings (the Jordan curves are simple), (b) no two edge curves share more than one point, and (c) no three edge curves share a common internal point.

For a vertex~$v$ of~$G$ and a drawing $\Gamma$ of $G$, let $x_{\Gamma}(v)$ and $y_{\Gamma}(v)$ denote the x- and y-coordinates of~$v$ 
in $\Gamma$, respectively; when $\Gamma$ is clear from the context, we may omit it and simply use the notation $x(v)$ and $y(v)$.  
A \emph{face} of $\Gamma$ is a region of the plane delimited by maximal uncrossed segments of the edges of $G$. 
The unique unbounded face of $\Gamma$ is its \emph{outer face}, the other faces are its \emph{internal faces}. 
An \emph{outer edge} is one incident to the outer face; all other edges are \emph{inner edges}. 
The \emph{rotation} of a vertex~$v$ in~$\Gamma$ is the counterclockwise cyclic order of the edges incident to~$v$. 
The \emph{rotation system} of~$\Gamma$ is the set of rotations of its vertices. 

The drawing $\Gamma$ is \emph{planar} if no two of its edges cross; it
is \emph{$k$-planar} if each edge is crossed at most $k$ times.  A
graph is \emph{($k$-)planar} if it admits a ($k$-)planar drawing; it
is \emph{outer ($k$-)planar} if it admits a ($k$-)planar drawing where
all vertices are incident to the outer face.
A \emph{planar embedding} $\mathcal{E}$ of a planar graph $G$ is an equivalence class of planar drawings of~$G$, namely those that have the same set of faces. Each face can be described as a sequence of edges and vertices of~$G$ which bound the corresponding region in the plane; each such sequence is a face of~$G$ in the embedding~$\mathcal{E}$. A planar embedding~$\mathcal{E}$ of a connected graph can also be described by specifying the rotation system and the outer face associated with any drawing of~$\mathcal{E}$.

Let $\Gamma$ be a non-planar drawing of a graph~$G$; the \emph{planarization} of~$\Gamma$ is the planar drawing~$\Gamma'$ of the \emph{planarized graph}~$G'$ obtained by replacing each crossing of $\Gamma$ with a dummy vertex. If $\Gamma$ is 1-planar, the planarization can be obtained as follows.
Let $uv$ and $wz$ be any two edges that cross in $\Gamma$; they are replaced in $\Gamma'$ by the edges $ux$, $xv$, $wx$ and $xz$, where $x$ is the dummy vertex. Two non-planar drawings of a graph $G$ have the same \emph{embedding} if their planarizations have the same planar embedding.
An embedding $\mathcal{E}$ of $G$ can also be described by specifying the planarized graph $G'$ and one of its planar embeddings. A planar graph with a given planar embedding is also called \emph{plane graph}. 
An \emph{outerplane graph} is a plane graph whose vertices are all incident to the outer face.
A \emph{fan} is a maximal outerpath that contains a vertex $c$ that is adjacent to
all other vertices; we call $c$ the \emph{central vertex} of the fan.
A \emph{$2$-tree} is a planar graph that can be reduced to an edge by iteratively removing a degree-two vertex that closes a $3$-cycle. A \emph{series-parallel graph} is a graph that can be augmented to a $2$-tree by adding edges (and no vertices).

A \emph{(simple, finite) directed graph} (\emph{digraph} for short) $G$ consists of a finite set $V(G)$ of \emph{vertices} and a finite set $E(G) \subseteq \{(u,v) \mid u,v \in V(G), u \ne v\}$ of ordered pairs of vertices, which are called \emph{edges}.  A \emph{source} (resp.\ \emph{sink}) of $G$ is a vertex with no incoming (resp.\ no outgoing) edges. A \emph{single-source graph} is a digraph with a single source and, possibly, multiple sinks. A digraph $G$ is an \emph{st-graph} if: (i) it is acyclic and (ii) it has a single source $s$ and a single sink $t$. 
An st-graph is a \emph{planar st-graph} if it admits a planar embedding with $s$ and $t$ on the outer face. 
We say that a drawing of a digraph $G$ is \emph{upward} if every (directed) edge~$(u,v)$ of~$G$ is mapped to a $y$-monotone Jordan arc with $y(u) < y(v)$.
Clearly, a digraph admits an upward drawing only if it does not contain a directed cycle.  Therefore, we assume for the rest of the paper that the input graph is a \emph{DAG}, a directed acyclic graph. 
Such a graph has a \emph{linear extension}, i.e., a vertex order $v_1,\dots,v_n$ such that, for any directed edge $(v_i,v_j)$, we have $i<j$. 
We say that a DAG is planar, outerplanar, or series-parallel if its underlying undirected graph is planar, outerplanar, or series-parallel, respectively.

Let $\Gamma$ be an upward drawing of a DAG $G$. By the upwardness, the
rotation system of $\Gamma$ is such that for every vertex $v$ of
$\Gamma$ the rotation of $v$ has only one maximal subsequence of
outgoing (incoming) edges. We call such a rotation system a \emph{bimodal rotation system}.
An \emph{upward embedding} of a DAG $G$ is an embedding of $G$ arising from an upward drawing; it naturally has a bimodal rotation system. 
The minimum $k$ such that a digraph $G$ admits an upward $k$-planar drawing is called its
\emph{upward local crossing number} and denoted by $\lcr(G)$.

For any positive integer $k$, we use $[k]$ as shorthand for
$\{1,2,\dots,k\}$.  A \emph{path-decomposition} of a graph $G$ is a
sequence $P=\left<X_1,\dots,X_\ell\right>$ of subsets of $V(G)$,
called \emph{buckets}, such that (1)~for each edge~$e$ of~$G$ there is
a bucket that contains both end vertices of $e$, and (2)~for every
vertex~$v$ of~$G$, the set of buckets that contain~$v$ form a
contiguous subsequence of~$P$. The \emph{width} of a
path-decomposition is one less than the size of the largest
bucket. The \emph{pathwidth} of the graph $G$ is the width of a path
decomposition of~$G$ with the smallest width.
\section{Lower Bounds}\label{sec:lower-bounds}

We start with a negative result that shows that even very special directed acyclic outerpaths may not admit upward-planar drawings, thus strengthening Papakostas' observation~(\citeyear{p-uptod-GD94}).

\begin{proposition}
  \label{prop:fan}
  Not every directed acyclic fan is upward-planar. 
\end{proposition}

\begin{proof}
  Consider the 7-vertex fan $F$ depicted in \cref{fig:non-upward-planar-fan-a}.
  Suppose for a contradiction that $F$ is upward planar, that is, $F$
  admits an upward planar drawing~$\Gamma$.
  Let $c$ be the central vertex of~$F$.
  We assume that $c$ is placed at the origin.  
  We say that a triangle of~$F$ is \emph{positive} (\emph{negative}, respectively) if the
  corresponding region of the plane in~$\Gamma$ contains the point
  $(\varepsilon,0)$ ($(-\varepsilon,0)$, respectively) for a
  sufficiently small value $\varepsilon>0$.
  The triangles that have one vertex below~$c$ and one vertex
  above~$c$ (namely $t_1=\triangle cv_1v_2$, $t_3=\triangle cv_3v_4$,
  and $t_5=\triangle cv_5v_6$) are either positive or negative.

\begin{figure}[tbh]
  \begin{subfigure}[b]{.23\textwidth}
    \centering
    \includegraphics[page=1]{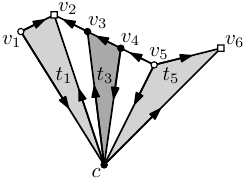}
    \subcaption{the fan $F$}
    \label{fig:non-upward-planar-fan-a}
  \end{subfigure}
  \hfill
  \begin{subfigure}[b]{.28\textwidth}
    \centering
    \includegraphics[page=2]{non-planar-fan}
    \subcaption{case: $t_5$ is contained in $t_1$}
    \label{fig:non-upward-planar-fan-b}
  \end{subfigure}
  \hfill
  \begin{subfigure}[b]{.34\textwidth}
    \centering
    \includegraphics[page=3]{non-planar-fan}
    \subcaption{%
    	a 1-planar upward drawing of $F$}
    \label{fig:non-upward-planar-fan-c}
  \end{subfigure}
  \caption{A directed acyclic fan $F$ that does not admit an upward planar
    drawing.}
  \label{fig:non-upward-planar-fan}
\end{figure}

  If both 
  $t_1$ and $t_5$
  are positive, then one must contain the other
  in~$\Gamma$, say, $t_1$ contains $t_5$; see
  \cref{fig:non-upward-planar-fan-b}.
  But then vertices~$v_3$ and~$v_4$ must also lie inside $t_1$.
  If both lie inside $t_5$, then the edge $(v_3,v_2)$ intersects an
  edge of~$t_5$.  If one of them lies inside~$t_5$ and one does not, then the edge $(v_4,v_3)$ intersects an
  edge of~$t_5$.  So both must lie outside $t_5$.
  But~$v_4$ lies on one hand above~$v_5$ and on the other hand below
  $c$ and, thus, below~$v_6$. So the edge $(v_4,c)$ intersects the
  edge $(v_5,v_6)$.  (If~$t_5$ contains~$t_1$, the edge $(c,v_3)$ intersects the edge $(v_1,v_2)$.)

  By symmetry, not both $t_1$ and $t_5$ can be negative, so exactly one
  of~$t_1$ and~$t_5$ must be negative, say, $t_1$; see
  \cref{fig:non-upward-planar-fan-c}.  Now first assume that $t_3$ is
  positive.  Due to edge $(v_3,v_2)$, vertex~$v_3$ must be
  outside~$t_5$, so $t_3$ cannot be inside~$t_5$.  On the other hand,
  $t_3$ cannot contain~$t_5$ because $v_4$ is above $v_5$.  Hence
  $t_3$ intersects~$t_5$.  Finally, assume that $t_3$ is negative.
  Due to edge $(v_5,v_4)$, vertex~$v_4$ must be
  outside~$t_1$, so $t_3$ cannot be inside~$t_1$.  On the other hand,
  $t_3$ cannot contain~$t_1$ because $v_3$ is below~$v_2$.  Hence
  $t_3$ intersects~$t_1$.
\end{proof}

By iteratively adding paths on every outer edge of an outerplanar but not upward-planar DAG, we can construct outerplanar DAGs with an unbounded upward local crossing number.

\begin{theorem}
    \label{thm:outer-unbounded}
    For each $\ell \ge 0$,
    there is a bipartite outerplanar DAG $G_\ell$ with $n_\ell = 8\cdot 3^\ell$ vertices, maximum degree $\Delta_\ell = 2\ell + 3$, and upward local crossing number greater than $\ell/6$, which is in $\Omega(\log n_\ell)$ and $\Omega(\Delta_\ell)$.
\end{theorem}

\begin{proof}
  The bipartite graph $G_0$ in \cref{subfig:papakostas} is not upward
  planar~\citep{p-uptod-GD94}. 
  For $\ell \geq 1$, we construct $G_\ell$ from $G_{\ell-1}$ by adding a $3$-edge path on every outer edge of
  the graph. \cref{subfig:G2}
  shows~$G_2$.  The maximum degree of $G_\ell$ is
  $\Delta_\ell = 2\ell + 3$.  The number of vertices is
  $n_\ell = 8 + \sum_{i=1}^\ell 8 \cdot 3^{i-1} \cdot 2 = 8 \cdot 3^\ell$.

\begin{figure}[ht]
    \begin{subfigure}{0.4\textwidth}
        \centering
        \includegraphics[page=3]{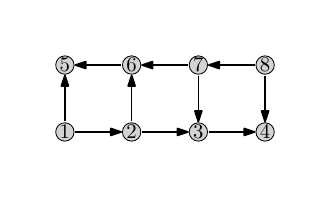}
        \subcaption{$G_0$}
        \label{subfig:papakostas}
    \end{subfigure}
    \hfill
    \begin{subfigure}{0.4\textwidth}
        \centering
	  \includegraphics[page=4]{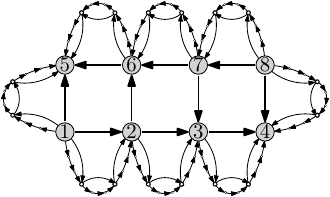}
        \subcaption{%
        	$G_2$}
        \label{subfig:G2}
    \end{subfigure}
    \caption{There is a family $(G_\ell)_{\ell \ge 0}$ of bipartite outerplanar graphs such that $G_\ell$ has $n_\ell$ vertices, maximum degree $\Delta_\ell$, and upward local crossing number in $\Omega(\Delta_\ell) \cap \Omega(\log n_\ell)$.}
    \label{fig:outerplanar}
\end{figure}
\begin{figure}[ht]
    \begin{subfigure}{0.4\textwidth}
        \centering
        \includegraphics[page=11]{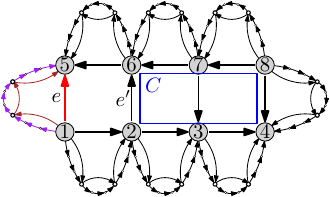}
        \subcaption{$e'$ inner edge of $G_0$}
        \label{subfig:inner}
    \end{subfigure}
    \hfill
    \begin{subfigure}{0.4\textwidth}
        \centering
	  \includegraphics[page=12]{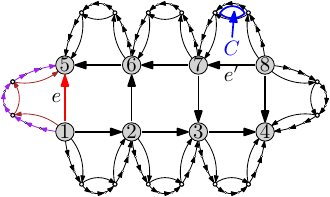}
        \subcaption{%
        	$e'$ outer edge of $G_0$}
        \label{subfig:outer}
    \end{subfigure}
    \caption{If $e$ crosses $e'$ an odd number of times then there is a cycle $C$ of length at most $6$ that is crossed by at least $\ell+1$ edge-disjoint paths, namely the edge $e$ and the $\ell$ paths added on top of $e$.}
    \label{fig:proofOuter}
\end{figure}

	Consider now an upward $k$-planar drawing $\Gamma$ of $G_\ell$ for some $k \le \ell$. 
    Since $G_0$ is not upward planar, there must be a pair of independent edges of $G_0$ that crosses an odd number of times in $\Gamma$.
Observe that $G_0$ has no upward planar drawing in which only the two inner edges cross an odd number of times, for otherwise the two cycles $\langle 1,2,6,5\rangle$ and $\langle 3,4,8,7\rangle$ would intersect an odd number of times, which is impossible. 
	Thus, in $\Gamma$ there must be an outer edge $e$ of $G_0$ that is crossed by an independent edge $e'$ of $G_0$ an odd number of times. 
  We choose $e'$ to be an outer edge of $G_0$, if possible.

We now determine a cycle $C$ of $G_\ell$ that is crossed by $e$ an odd
number of times and does also not contain any end vertex of $e$. If
$e'$ is an inner edge, then we take the outer path $P$ of~$G_0$ that
connects the ends of $e'$ and does not contain $e$; this is not
crossed by $e$ due to our choice of $e'$.
Moreover, since $e'$ and $e$ are independent, it follows that $P$ and $e$ do not share a vertex.
Let $C$ be the concatenation of~$P$ and~$e'$. 
Note that the length of~$C$ is at most six; see \cref{subfig:inner}.

If $e'$ is an outer edge of $G_0$, we do the following: We start
with the path $P$ of length three that was added for $e'$. Since $e$ is an edge of $G_0$ and $e$ and $e'$ are independent, it follows that $P$ and $e$ do not share a vertex. If $P$
contains an edge that is crossed an odd number of times by $e$ then 
we continue with such an edge instead of $e'$. More precisely, let
$e_1 = e'$ and initialize $i=1$. Let $P_i$ be the path of length three
that was added for $e_i$. While $P_i$ contains an edge that is crossed
an odd number of times by $e$, let $e_{i+1}$ be such an edge, let
$P_{i+1}$ be the path of length three that was added for $e_i$, and
increase $i$ by one.  Since $e$ is crossed at most $k$ times, this
process stops at some $i\leq k$.
Let $C$ be the cycle that is composed of
$P_i$ and $e_i$. See \cref{subfig:outer}. In this case $C$ has length four. Moreover $C$ shares at most the end vertices of $e'$ with $G_0$. Thus, since $e$ is an edge of $G_0$ and $e$ and $e'$ are independent, it follows that $C$ and $e$ do not share a vertex.

    Cycle $C$ might cross itself. However, it divides the plane into cells.
    Since $e$ crosses $C$ an odd number of times, it follows that the end vertices of $e$ must be in different cells of the plane. This means that not only $e$ but also the $\ell$ edge-disjoint paths that were added on top of $e$ have to cross $C$. Observe that none of these paths contains a vertex of $C$. But $C$ contains at most six edges, each of which can be crossed at most $k$ times.   
    This is impossible if $\ell \geq 6k$.  Hence, if there is an upward $k$-planar drawing then  $\ell < 6k$, which means that $k > \ell/6$.
\end{proof}

We now show that if we expand the graph class beyond outerplanar graphs, then we get a lower bound on the upward local crossing number that is even linear in the number of vertices.   
The graphs in our construction have pathwidth~2, as opposed to the graphs in \cref{thm:outer-unbounded} whose pathwidth is logarithmic.  
Observe that a \emph{caterpillar}, i.e., a tree that can be reduced to a path by removing all degree-1 vertices, has pathwidth 1, and that the pathwidth can increase by at most 1 if we add a vertex with some incident edges
or subdivide some edges.

\begin{theorem}
  \label{thm:unbounded_pw2}
  For every $k\ge 1$, there exists a (planar) DAG with $\Theta(k)$ vertices,
  maximum degree in $\Theta(k)$, and pathwidth~2 that does not admit
  an upward $k$-planar drawing.
\end{theorem}

\begin{proof}
Let $G_k$ be the graph consisting of the four vertices $a$, $b_1$, $b_2$, and~$c$ and the following set of edges and degree-2 vertices
(see also \cref{fig:pathwidth2}):
\begin{itemize}
\item edges $(a,b_1)$ and $(a,b_2)$;
\item for $i \in [2]$ and $j \in [3k+1]$, a \emph{through-vertex at $b_i$}, i.e., a vertex $d_i^{(j)}$ and 
	edges $(b_i,d_i^{(j)})$ and $(d_i^{(j)},c)$;
\item for $j \in [6k+1]$, a \emph{source below $a$}, i.e., a vertex $s^{(j)}$ and 
	edges $(s^{(j)},a)$ and $(s^{(j)},c)$;
\item for $i \in [2]$ and $j \in [4k+1]$, a \emph{sink above $b_i$}, i.e., a vertex $t_i^{(j)}$ and 
	edges $(b_i,t_i^{(j)})$ and $(c,t_i^{(j)})$.
\end{itemize}
Clearly, $G_k$ has $\mathcal O(k)$ vertices, and  pathwidth 2, since $G-c$ is a caterpillar and has pathwidth~1.

\begin{figure}[ht]
  \begin{subfigure}[b]{.48\textwidth}
    \includegraphics[page=1]{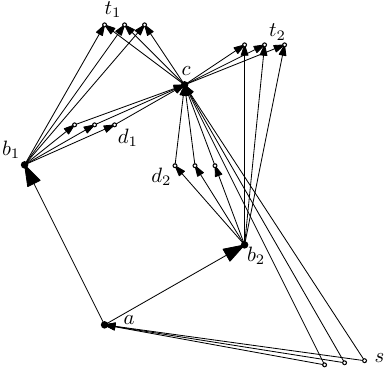}
    \subcaption{}
  \end{subfigure}
  \hfill
  \begin{subfigure}[b]{.48\textwidth}
    \includegraphics[page=2]{pathwidth2}
    \subcaption{}
  \end{subfigure}
  
  \caption{A graph of pathwidth 2 (drawn upward) that does not have an
    upward $k$-planar drawing.  (a)~We only show three of the
    $\Theta(k)$ vertices of each group.
    (b)~Cycles $C_a$ and $C_b$.
  }
  \label{fig:pathwidth2}
\end{figure}

Assume that there was an upward $k$-planar drawing $\Gamma$ of $G_k$.
Up to renaming, we may assume that $y(b_2)\leq y(b_1)$. 
Delete all but one of the through-vertices at $b_1$ from the drawing; in what  
follows, we write $d_1$ for the one that we keep (it does not matter which one).   

Among the $3k+1$ through-vertices $d_2^{(j)}$ at $b_2$, there exists at least one for which the path $\langle b_2, d_2^{(j)}, c \rangle$
crosses none of the three edges in the path $\langle a, b_1, d_1, c \rangle$, for otherwise there would be an edge with more
than $k$ crossings.   Delete all other through-vertices at $b_2$; in what follows we write $d_2$ for the one
that we keep. Let $a'$
be the topmost intersection point of $(a,b_1)$ and $(a,b_2)$ (possibly $a'=a$).
Since $y(a) \leq y(a') <y(b_2)\leq y(b_1)$
the curve $C_b$ formed by the two directed paths $\langle a',b_i,d_i,c \rangle$ (for $i \in [2]$) is
drawn without crossing in $\Gamma$.

Curve $C_b$ uses six edges, therefore among the $6k+1$ sources below $a$, there exists one, call it $s$, for which edge $(s,c)$ crosses no edge of $C_b$.   Since $y(s)<y(a)$, vertex $s$ is outside
$C_b$, and so the entire edge $(s,c)$ is outside $C_b$, except at the endpoint $c$.    In particular,
among the three edges $(d_1,c)$, $(d_2,c)$, and $(s,c)$ that are incoming at $c$,  edge $(s,c)$
is either leftmost or rightmost (but cannot be the middle one).    We assume here that $(s,c)$
is rightmost, the other case is symmetric.   Write $\{p,q\}=\{1,2\}$ such that the left-to-right order of incoming edges at $c$  is $(d_p,c)$, $(d_q,c)$, $(s,c)$.  In \cref{fig:pathwidth2}, we have $p=1$ and $q=2$.

Edge $(s,a)$ is also outside $C_b$, except perhaps at endpoint $a$, since it uses smaller $y$-coordinates.
Let $s'$ be the topmost intersection point of $(s,a)$ and $(s,c)$.
Then there are no crossings in the curve $C_a$ formed by the directed paths $\langle s', a, b_p,  d_p, c \rangle$
and $\langle s',c \rangle$.  By our choice of $p$ and $q$, vertex $d_q$ is {\em inside} $C_a$, and so is the entire path
$\langle a', b_q, d_q, c \rangle$, except at the ends since it is part of $C_b$.   In
particular, $b_q$ is inside $C_a$, whereas, for $j \in [4k+1]$, $t_q^{(j)}$ is outside $C_a$ due to 
$y(c)<y(t_q^{(j)})$.  It follows that one of the four edges $(a,b_p)$, $(b_p,d_p)$, $(d_p,c)$ and $(s,c)$
must be crossed at least $k+1$ times by edges from $b_q$ to the sinks above it.  Thus, the drawing was not $k$-planar, a contradiction.
\end{proof}

The graphs that we constructed in the proof of
\cref{thm:unbounded_pw2} are not bipartite, but one can make them
bipartite by subdividing all edges once.  This at best cuts the local 
crossing number in half, increases the pathwidth by at most~1, and
yields the following result.

\begin{corollary}\label{cor:BipartitePathwidth}
    There is a family of bipartite (planar)
    DAGs of constant pathwidth whose upward local crossing number is linear in the
    number of vertices.
\end{corollary}

So far we needed graphs of unbounded maximum degree in order to
enforce unbounded upward local crossing number. We now show that, intrinsically, this is not necessary.

\begin{proposition}\label{obs:cubic}
  There are cubic DAGs whose upward local
  crossing number is at least linear in the number of vertices.
\end{proposition}

\begin{proof}
  The crossing number of a random cubic graph with $n$ vertices is
  expected to be at least $cn^2$ for some absolute constant $c>0$
  \citep*{dujmovic_etal:scg08}, and thus there exist graphs yielding
  this bound.  By the pigeon-hole principle, such a graph contains an
  edge with $\Omega(n)$ crossings among its $\Theta(n)$ edges.  Impose
  arbitrary acyclic edge directions.
\end{proof}

\section{Upper Bounds}\label{sec:upper-bounds}

The \emph{bandwidth} $\bw(G)$ of an undirected graph $G$ is the smallest positive integer~$k$
such that there is a labeling of the vertices by distinct numbers $1, \dots, n$ for which the labels
of every pair of adjacent vertices differ by at most $k$.

\begin{theorem}\label{thm:bandwidth}
  The upward local crossing number of a DAG~$G$ with maximum
  degree~$\Delta$ is at most 
  $\Delta \cdot (\bw(G) - 2) \le 2 \bw(G) (\bw(G) - 2)$, 
  so it is in $\mathcal{O}(\Delta\cdot\bw(G))\subseteq{\mathcal O}(\bw(G)^2)$.
\end{theorem}

\begin{proof}
  Observe that the maximum degree $\Delta$ of a graph $G$ is bounded
  in terms of the bandwidth of~$G$; namely, $\Delta \le 2 \bw(G)$.
  Consider a linear extension of~$G$.  For every vertex~$v$ of~$G$,
  let $y(v)$ be its index in the extension.  Now consider a labeling
  of~$G$ corresponding to the bandwidth.  For every vertex~$v$ of~$G$,
  let $x(v)$ be its
  label.
 
  Construct a drawing of $G$ by first placing every vertex $v$ at the
  point $(x(v),y(v))$ and by then perturbing vertices slightly so that
  the points are in general position.  Adjacent vertices are connected
  via straight-line segments.
	
  It is easy to see that the drawing is upward since it is consistent
  with the linear extension.  Consider an arbitrary edge $(u,v)$ and assume, without loss of generality, that
  $x(u)<x(v)$.  Let $\ell = x(v)-x(u)$.
The length of any edge in x-direction is bounded by $\bw(G)$; hence $\ell\le\bw(G)$.
Edge $(u,v)$ may be crossed (a)~by edges that have at least one incident vertex $w$ with $x(u)<x(w)<x(v)$, or (b)~by edges between two vertices $w$ and $w'$ with $x(w)<x(u)$ and $x(v)<x(w')$.
Thus, there are at most $\ell-1$ and $\bw(G)-\ell-1$ possible choices for $w$ in the two scenarios, respectively. Since each vertex is incident to at most $\Delta$ edges,
$(u,v)$ can be crossed at most $\Delta(\bw(G)-2)$ times.
\end{proof}	

For some graph classes, sublinear bounds on the bandwidth are known
\citep*[see, e.g.,][]{BPTW2010,FG2003,Wood2006}.  This gives upper
bounds on the upward local crossing number of many graph classes.  We
list only a few:

\begin{corollary}
  The following classes of DAGs have sublinear upward local crossing number:
  \begin{itemize}
  \item Square $k \times k$ grids have bandwidth $\Theta(k)$ and
    $\Delta=4$, hence their upward local crossing number is in
    $\mathcal O(k) = \mathcal O(\sqrt n)$.
		
  \item Planar graphs of maximum
    degree~$\Delta$
    have bandwidth ${\mathcal O}({n}/{\log_\Delta n})$ \citep{BPTW2010},
    hence their upward local crossing number is in
    $\mathcal O({n \cdot \Delta}/{\log_{\Delta} n})$. Observe that ${n \cdot \Delta}/{\log_{\Delta} n} \in o(n)$ if $\Delta \in o({\log n}/{\log \log n})$.
  \end{itemize}
\end{corollary}	

We now complement the negative result in \cref{prop:fan} by showing that
every directed acyclic outerpath allows an upward 2-planar drawing. We
start with a technical lemma on fans.

  \begin{figure}[htb]
    \begin{minipage}[b]{.27\textwidth}
      \centering
      \includegraphics[page=1]{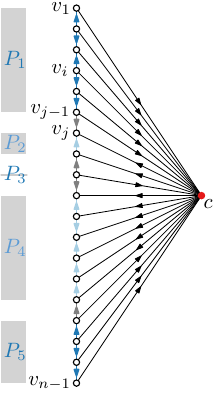}
      \subcaption{a directed fan~$G$}
      \label{fig:fan-input}
    \end{minipage}
    \hfill
    \begin{minipage}[b]{.59\textwidth}
      \centering
      \includegraphics[page=2]{fan}
      \subcaption{upward drawing of $G$ with at most
        two crossings per edge}
      \label{fig:fan-output}
    \end{minipage}
    \caption{Constructing upward $2$-planar drawings of fans according to \cref{lemma:fan}.
    }
    \label{fig:fan}
  \end{figure}

\begin{lemma}
  \label{lemma:fan}
  Let $c$ be the central vertex of a directed acyclic fan $G$, and let
  $P=\langle v_1,\dots,v_{n-1} \rangle$ be the path of the remaining
  vertices in $G$.  Let $P_1,\dots,P_k$ be an ordered partition of $P$
  into maximal subpaths such that, for every $i \in [k]$, the edges
  between $P_i$ and $c$ either are all directed towards $c$ or are all
  directed away from~$c$.  Then there is an upward 2-planar drawing of
  $G$ with the following properties:
  \begin{enumerate}
  \item no edge incident to $c$ is crossed;
  \item %
  the central vertex $c$ and
    $v_{n-1}$ have x-coordinate $n-1$, and the x-coordinates of
    $v_1,\dots,v_{n-2}$ are pairwise distinct values within $\{1,\dots,n-2\}$;
  \item for all edges all $x$-coordinates of the curves are at most
    $n-1$; all edges incident to $c$ and all edges of the subpaths
    $P_1,\dots,P_k$ are in the vertical strip between $1$ and $n-1$;
  \item\label{itm:crossing} if $P_1$ is a directed path, then the edge
    between $P_1$ and $P_2$ is crossed at most once.
  \end{enumerate}
\end{lemma}

\begin{proof}
  We place~$c$ at $(n-1,0)$; then we place $v_1, v_2, \dots, v_{n-1}$
  above or below~$c$ depending on the direction of the edges that
  connect them to~$c$; see \cref{fig:fan} for an example.
    
  For each $\ell \in [n-2]$, we keep the invariant that, when we
  place~$v_\ell$, the \emph{leftward} ray, that is, the one in direction $1 \choose 0$
  from~$v_\ell$, reaches the outer face of the current drawing after crossing 
  at most one other edge, and that this edge is currently crossed by at most one edge.

  In order to choose appropriate y-coordinates, we maintain two values
  $y_{\min}$ and $y_{\max}$ indicating the minimum and maximum
  y-coordinate of any so far drawn vertex.  Consider now a subpath
  $P' \in \{P_1,\dots,P_k\}$.  Let $v_h$ be the first and let
  $v_{j-1}$ %
  be the last vertex of $P'$, i.e.,
  $P'=\langle v_h, v_{h+1}, \dots, v_{j-1} \rangle$.  We describe in
  detail the case that the edge from $v_h,\dots,v_{j-1}$ to~$c$ are
  directed towards~$c$ that is, $v_h$ must lie below~$c$.  The other
  case is symmetric.  We place~$v_h$ at x-coordinate $h$ and with a
  y-coordinate sufficiently below $y_{\min}$.  If $h=j-1$ we are done.
  
  We now consider the case that $j=n$ or
  $(v_{j-1},v_{j-2}) \in E(G)$; see, for example, path $P_5$ in \cref{fig:fan-output}.
  In this case, we place $v_{h+1},\dots,v_{j-1}$ using x-coordinates
  $h+1,\dots,j-1$, going up and down as needed but remaining below the
  x-axis. The edges are drawn such that all vertices of $P'$ remain on
  the outer face of the drawing.  I.e., if we use straight-line edges,
  then, for $i \in [n-2]$, the slope of $v_{i}v_{i+1}$ must be less
  than the slope of $v_{i}c$. Since we go towards~$c$, we can draw
  $P'$ and the edges that connect $v_1,v_2,\dots,v_{n-1}$ to~$c$
  without any crossings.
  
  If $j\neq n$ and $(v_{j-2},v_{j-1}) \in E(G)$, then
  let~$i \in \{h,\dots,j-1\}$ be the smallest index such that the
  subpath $\langle v_i, v_{i+1}, \dots, v_{j-1}\rangle$ is directed.
  See, for example, that last part of $P_1$ in \cref{fig:fan-output}.
  In that case, we place $v_{h+1},\dots,v_{i-1}$ at x-coordinates
  $h+1,\dots,i-1$, going slightly up and down as in the case described
  above.  Let $y_{\min}$ be the smallest among the y-coordinates of
  all points placed so far.

  Then we place~$v_i,v_{i+1},\dots,v_{j-1}$ in reverse order, i.e., at
  x-coordinates~$j-1,j-2,\dots,i$.
  For the y-coordinates, we choose $y(v_i)=y_{\min}-\gamma$ and
  $y(v_{j-1})=y_{\min}-\varepsilon$ for some (large) $\gamma>0$ and
  (small) $\varepsilon>0$ with the following properties: if $i > h$ then $v_{j-1}$ lies inside the triangle
  $\triangle v_{i-1}v_ic$ (pale yellow in \cref{fig:fan-output}) and if
  $i = h$ then $v_{j-1}$ lies inside the triangle $\triangle ov_ic$, where $o = (0,0)$. 
  Draw $v_{i+1},\dots,v_{j-2}$ on the segment
  $\overline{v_iv_{j-1}}$. This fulfills the invariant: if $i > h$ then the vertex $v_{j-1}$
  can reach the outer face via the edge~$(v_i,v_{i-1})$ which was not
  crossed so far. If $i=h$ then $v_{j-1}$ is on the outer face if
  $P'=P_1$, otherwise $v_{j-1}$ can reach the current outer face by crossing the edge~$(v_h,v_{h-1})$.
  Observe that, by our invariant, $(v_h,v_{h-1})$ might have
    crossed one edge in order to reach the outer face. While drawing $P'$, we do not cross $(v_h,v_{h-1})$ again. So the potential crossing with  $(v_{j-1},v_j)$ is at most the second crossing of $(v_h,v_{h-1})$, and $(v_h,v_{h-1})$ will not be crossed later.

  Moreover, when we draw the next maximal subpath, we place $v_j$
  at $(j,y_{\max}+1)$, i.e., in particular in the outer face of the
  current drawing.  The edge from $v_{j-1}$ to $v_j$ must be directed
  towards $v_j$ since the orientation is acylic.  Thus, we can draw
  the edge between $v_{j-1}$ and $v_j$ upward with at most one
  crossing, causing at most a second crossing on $(v_h,v_{h-1})$.
\end{proof}

\begin{figure}[htb]
  \begin{minipage}[b]{.32\textwidth}
    \centering
    \includegraphics[page=3]{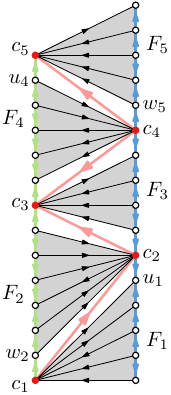}
    \subcaption{a directed outerpath $G'$}
    \label{fig:outerpath-input}
  \end{minipage}
  \hfill
  \begin{minipage}[b]{.61\textwidth}
    \includegraphics[page=4]{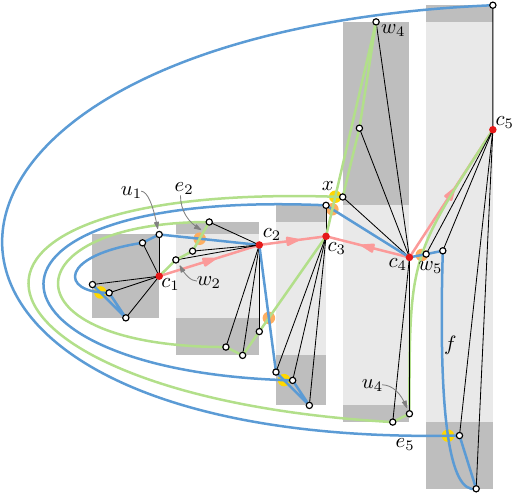}
    \subcaption{upward drawing of $G'$ with at most
      two crossings per edge}
    \label{fig:outerpath-output}
  \end{minipage}
  \caption{Example in- and output of our drawing algorithm (edge
    crossings due to \cref{lemma:fan} are highlighted in yellow; other
    edge crossings are highlighted in orange).
    }
  \label{fig:outerpath}
\end{figure}

\medskip

Now we describe our construction for general outerpaths; see \cref{fig:outerpath}.
  
\begin{theorem}
  \label{thm:outerpath}
  Every directed acyclic outerpath admits an upward 2-planar drawing.
\end{theorem}

\begin{proof}
  Without loss of generality, we can assume that the given outerpath
  is maximal: if the outerpath has interior faces that are not
  triangles, we temporarily triangulate them using additional edges,
  which we direct such that they do not induce directed cycles and
  which we remove after drawing the maximal outerpath.

  Let $G'$ be such a graph; see \cref{fig:outerpath-input}.  Let
  $c_1,c_2,\dots,c_k$ be the vertices of degree at least~4 in~$G'$
  (marked red in \cref{fig:outerpath}).  These vertices form a path
  (light red in \cref{fig:outerpath}); let them be numbered along this
  path, which we call the \emph{backbone} of~$G'$.  We assign every
  vertex~$v$ that does not lie on the backbone to a neighboring
  backbone vertex; if $v$ is incident to an inner edge, we assign~$v$
  to the other endpoint of that edge.  Otherwise $v$ has degree~2 and
  is incident to a unique backbone vertex via an outer edge, and we
  assign~$v$ to this backbone vertex.  For $i \in [k]$, backbone
  vertex $c_i$ induces, together with the vertices assigned to it, a
  fan $F_i$.

  We draw the backbone in an x-monotone fashion.  We start by drawing
  $F_1$ with the algorithm for drawing a fan as detailed in the proof
  of \cref{lemma:fan}; see the leftmost gray box in
  \cref{fig:outerpath-output}.  Then, for $i \in \{2,\dots,k\}$, we set
  $x(c_{i})$ to $x(c_{i-1})$ plus the number of inner edges incident
  to~$c_{i}$ and we set $y(c_i)$ depending (i)~on the y-coordinates
  of the two neighbors of $c_i$ that have already been drawn
  ($c_{i-1}$ and the common neighbor $u_{i-1}$ of $c_{i-1}$ and $c_i$ 
  in $F_{i-1}$) and 
  (ii)~on the directions of the edges that connect these vertices
  to~$c_i$; see, for example, the placement of~$c_5$ in
  \cref{fig:outerpath-output}.  Then we draw~$F_i$ with respect to
  the position of~$c_i$, again using the algorithm from the proof
  of \cref{lemma:fan} with the following modifications. 
  In general, vertices
  in~$F_i$ that are adjacent to~$c_i$ via an edge directed
  towards~$c_i$ (resp.\ from~$c_i$) are placed below (resp.\
  above) all vertices in the drawings of $F_1,\dots,F_i$; see the dark
  gray boxes below (resp.\ above) $c_2,\dots,c_5$ in
  \cref{fig:outerpath-output}.  If an edge of~$F_i$ connects two
  neighbors of~$c_i$ one of which lies above~$c_i$ and one of
  which lies below~$c_i$, then we route this edge to the left of
  all drawings of $F_1,\dots,F_{i-1}$.
    
  An exception to this rule occurs if $c_i$ and the common neighbor~$w_i$ 
  of~$c_{i-1}$ and~$c_i$ in~$F_i$ must be both above or both 
  below~$c_{i-1}$ due to the directions of the corresponding edges.  
  Let $u_{i-1}$ be the common neighbor of~$c_{i-1}$ and~$c_i$ in~$F_{i-1}$.
  We assume, without loss of generality, that $c_i$ is above~$c_{i-1}$. 
  Let $P_1$ and $P_2$ be the first and second maximal subpath from 
  \cref{lemma:fan} applied to~$F_i$, and let~$e_i$ be the edge 
  connecting~$P_1$ and~$P_2$.  We distinguish two subcases.  

  If $P_1$ is a directed path leaving~$w_i$, then we 
  draw~$P_1$ above the edge $c_{i-1}c_i$ and we draw the edge~$e_i$ 
  straight, without going around all drawings of $F_1,\dots,F_{i-1}$.
  In this case $e_i$ is directed from~$P_1$ to~$P_2$.
  Hence, $e_i$ crosses the edge~$u_{i-1}c_i$ if $u_{i-1}c_i$ 
  is directed from~$c_i$ to~$u_{i-1}$; see the situation for~$c_2$ 
  in \cref{fig:outerpath-output}.  Note that by Property~\ref{itm:crossing} of \cref{lemma:fan}, $e_i$ may receive
  at most a second crossing when we draw the remainder of~$F_i$ in the 
  usual way.

  Otherwise, that is, if $P_1$ contains an edge directed 
  towards the first
  endpoint~$w_i$ of~$P_1$, 
  let~$f$ be the first such edge.  We then place the 
  part of~$P_1$ up to the first endpoint of~$f$ below the 
  edge~$c_{i-1}c_i$; see~$w_5$ and~$f$ in \cref{fig:outerpath-output}. 
  If the edge $u_{i-1}c_i$ is directed towards~$c_i$, we draw it between~$w_i$ 
  and the edge $c_{i-1}c_i$.  Then it crosses the edge~$c_{i-1}w_i$
  but no other edge.  We place the second endpoint of~$f$ 
  below all vertices in $V(F_1) \cup \dots \cup V(F_{i-1})$ and 
  continue with the remainder of~$F_i$ as usual.  
  
  In any case, if $1<i<k$, then
  the last vertex~$u_{i-1}$ of~$F_{i-1}$ is connected to~$c_i$
  and~$c_{i-1}$ is connected to the first vertex~$w_i$ in~$F_i$.
  These two edges may cross each other; see the crossings 
  highlighted in orange in \cref{fig:outerpath-output}.  If the 
  edge~$c_{i-1} w_i$ goes, say, up but the following outer edges go down 
  until a vertex $v_k$ below~$c_i$ is reached, then the edge~$c_{i-1}w_i$ 
  may be crossed a second time by the edge~$v_{k-1}v_k$; see
  the crossing labeled~$x$ on the edge~$c_3 w_4$ in
  \cref{fig:outerpath-output}. Observe that in this case the path $P_1$ is the directed path from $w_i$ to $v_{k-1}$. Thus, due to Property~\ref{itm:crossing}
  of \Cref{lemma:fan},
  edge~$v_{k-1}v_k$ had been crossed at most once within its fan. 
  Also $c_{i-1}w_i$ cannot have a third crossing. 
  Thus, all edges are crossed at most twice.
\end{proof}

One can argue
that every maximal pathwidth-2 graph can be generated from a
maximal outerpath by connecting some pairs of adjacent vertices using
an arbitrary number of (new) paths of length~2.  In spite of the
simplicity of this operation, we cannot hope to generalize the above
result to pathwidth-2 graphs; see the linear lower bound on the
upward local crossing number for such graphs stated in
\cref{thm:unbounded_pw2}.

\section{Testing Upward 1-Planarity}
\label{sse:NP-complete}

Here, we prove that upward 1-planarity testing is \NP-complete even for structurally simple  DAGs, both when a bimodal rotation system is fixed and when it is not fixed. 

We start with a definition.
Let  $G_1$ and $G_2$ be any two st--digraphs. Let $s_i$ be the source and $t_i$ be the sink of $G_i$, with $i=1,2$. Let $G$ be a digraph that contains both $G_1$ and $G_2$ as induced subgraphs. Let $\Gamma$ be a  drawing of $G$ and let $\Gamma_{1,2}$ be the drawing obtained by restricting $\Gamma$ to the vertices and edges of  $G_1 \cup G_2$. We say that $G_1$ and $G_2$ \emph{fully cross}
in $\Gamma$ if in $\Gamma_{1,2}$ every $s_1t_1$-path (i.e., a path directed from $s_1$ to $t_1$) crosses every $s_2t_2$-path. 
See \cref{fig:st-graphs-non-crossing,fig:st-graphs-crossing-v1} for examples of st-digraph $G_1$ and $G_2$ that do not fully cross or fully cross in a drawing~of~$\Gamma_{1,2}$,~respectively. %

\begin{figure}[ht]
	\begin{subfigure}[b]{0.3\textwidth}
		\centering
		\includegraphics[page=1]{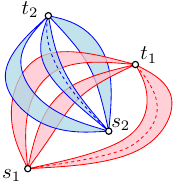}
		\subcaption{}
            \label{fig:st-graphs-non-crossing}
	\end{subfigure}
        \hfill
	\begin{subfigure}[b]{0.3\textwidth}
		\centering
		\includegraphics[page=2]{st-graph-crossing-modified-again.pdf}
		\subcaption{}
		\label{fig:st-graphs-crossing-v1}
	\end{subfigure} 
        \hfill
	\begin{subfigure}[b]{0.3\textwidth}
		\centering
		\includegraphics[page=3]{st-graph-crossing-modified-again.pdf}
		\subcaption{}
		\label{fig:st-graphs-crossing-v2}
	\end{subfigure} 
        \hfill
	\caption{Illustrations for the definition of fully crossing st-subgraphs. (a) and (b) Two st-digraphs $G_1$ and $G_2$ that do not fully cross, as witnessed by the two non-crossing dashed paths. (c) Two st-digraphs $G_1$ and $G_2$ that fully cross. 
 }
	\label{fig:ga-gb-crossing}
\end{figure}

We now define a few gadgets; all of them are planar st-graphs.
For positive integers $b$ and $q$, let a \emph{$(b,q)$-parallel} be the parallel composition of $b$ oriented paths each consisting of $q$ edges; see \cref{fig:3-4-parallel}.
For a positive integer $p$, let a \emph{$(p)$-gate} be the parallel composition of an oriented edge and a $(p-1,2)$-parallel; see \cref{fig:4-gate}.
For positive integers $h$, $q$, and $a$, let an \emph{$(h,q,a)$-chain} consist of a series of $h$ $(q)$-gates, followed by exactly one $(a)$-gate, followed again by $h$ $(q)$-gates; see \cref{fig:253-chain}.

\begin{figure}[ht]
	\begin{subfigure}[b]{0.20\textwidth}
		\centering
		\includegraphics[page=1]{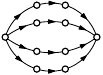}
		\subcaption{a $(4,3)$-parallel}
            \label{fig:3-4-parallel}
	\end{subfigure}
        \hfill
	\begin{subfigure}[b]{0.13\textwidth}
		\centering
		\includegraphics[page=2]{reduction-lemma.pdf}
		\subcaption{a $(4)$-gate}
		\label{fig:4-gate}
	\end{subfigure} 
        \hfill
	\begin{subfigure}[b]{0.61\textwidth}
		\centering
		\includegraphics[page=3]{reduction-lemma.pdf}
		\subcaption{a $(2,5,3)$-chain}
		\label{fig:253-chain}
	\end{subfigure}
	\caption{Illustrations for the gadgets used in the construction of $G_A$ and of $G_B$.}
	\label{fig:reduction-lemma}
\end{figure}

An instance of \textsc{3-Partition} is a multiset $I=\{a_1,a_2,a_3, \dots, a_k\}$ of positive integers such that $b = k/3$ is an integer and $\sum_{i=1}^k a_i = W \cdot b$, with $W$ integer. 
The \textsc{3-Partition} problem asks if there exists a partition of the set $I$ into $b$ 3-element subsets such that the sum of the elements in each subset is exactly~$W$. 
Since \textsc{3-Partition} is strongly \NP-hard~\citep{garey1979computers}, we may assume that $W$ is bounded by a polynomial in~$b$. 

We associate with a given instance $I$ of \textsc{3-Partition} two planar st-graphs $G_A$ and $G_B$ defined as follows. 
Digraph $G_A$ is the parallel composition of $(b-1,W+1,a_i)$-chains, one for every $i \in \{1, \dots, k\}$. 
Digraph $G_B$ is a $(b,q)$-parallel, with $q = W + (k-3)(W+1)$. Note that the underlying undirected graphs of both $G_A$ and $G_B$ are series-parallel.

\renewcommand{\topfraction}{0.85}
\renewcommand{\bottomfraction}{0.85}
\renewcommand{\textfraction}{0.15}
\renewcommand{\floatpagefraction}{0.85}

\begin{theorem}
  \label{le:3-partition}
  Let $I$ be an instance of \textsc{3-Partition} and let $G_A$ and
  $G_B$ be the two planar st-graphs associated with $I$. Assume there
  exists a digraph $G$ containing $G_A$ and $G_B$ as subgraphs with
  the following two properties: (i) if $G$ is upward 1-planar then
  $G_A$ fully crosses $G_B$ in every upward 1-planar drawing of $G$;
  (ii) if there exists an upward 1-planar drawing of the union of
  $G_A$ and $G_B$ in which $G_A$ and $G_B$ fully cross, then there
  exists an upward 1-planar drawing of $G$.  Then the digraph $G$ is
  upward 1-planar if and only if $I$ admits a solution.
\end{theorem}

\begin{proof}
  Assume that $G$ is upward 1-planar and let $\Gamma$ be any upward
  1-planar drawing of $G$. We prove that $\Gamma$ provides a solution
  of instance $I$ of \textsc{3-Partition}.  By hypothesis~(i), $G_A$
  and~$G_B$ fully cross in~$\Gamma$.
  Observe that only one path among the $b$ paths of the $(b,q)$-parallel
  $G_B$ can traverse the $(a_i)$-gate
  of a $(W+1,b-1,a_i)$-chain of $G_A$, where $a_i$, with
  $i = 1, \dots, k$, is an element of the instance
  $I=\{a_1,a_2,a_3,\dots a_k\}$ of \textsc{3-Partition}, as otherwise
  the directed edge connecting the source and the sink of the
  $(a_i)$-gate would be traversed more than once.  Since $G_A$ and
  $G_B$ fully cross in $\Gamma$, by definition every path of $G_B$
  crosses all the $(W+1,b-1,a_i)$-chains of $G_A$.
  In particular, every path of $G_B$ must cross at least three
  $(a_i)$-gates. Indeed, if a path $\pi$ of $G_B$ crossed less than
  three $(a_i)$-gates of $G_A$, it would cross at least $k-2$ of the
  $(W+1)$-gates in the $(W+1,b-1,a_i)$-chains of $G_A$. In order to have at most one
  crossing per edge, path $\pi$ should have at least $(k-2)(W+1)$
  edges; however, by construction, $\pi$ has
  $W+(k-3)(W+1)= (k-2)(W+1) - 1$ edges. Also, observe that by
  definition of $G_A$, if one path of $G_B$ crossed more than three
  $(a_i)$-gates, some other path of $G_B$ should cross at most two
  $(a_i)$-gates. Therefore, every path $\pi$ of $G_B$ must cross
  exactly three $(a_i)$-gates and $k-3$ of the $(W+1)$-gates in
  $\Gamma$. Since every path $\pi$ of $G_B$ uses $(k-3)(W+1)$ of its
  edges to cross $k-3$ of the $(W+1)$-gates in $\Gamma$, $\pi$ can cross
  at most $W$ edges of the $(a_i)$-gates. Since
  $\sum_{i=1}^k a_i = W \cdot b$, the number of crossings of each
  $\pi$ with its three $(a_i)$-gates must be exactly $W$. It follows
  that if $G$ has an upward 1-planar drawing then the instance $I$ of
  \textsc{3-Partition} admits a solution.
	
  Conversely, assume that the instance $I$ of \textsc{3-Partition}
  admits a solution. In order to prove that $G$ admits an upward
  1-planar drawing, by hypothesis (ii) it suffices to construct an
  upward 1-planar drawing $\Gamma_{AB}$ of $G_A$ and $G_B$ where $G_A$
  and $G_B$ fully cross.  We enumerate the paths of $G_B$ as
  $\pi_1, \pi_2, \dots, \pi_b$.  We also enumerate the $2(b-1)$ many
  $(W+1)$-gates and the
  $(a_i)$-gate
  of a $(b-1,W+1,a_i)$-chain as $g_{i,1}, g_{i,2}, \cdots, g_{i,2(b-1)}$, where $g_{i,1}$ contains
  the source~$s_A$ of~$G_A$ and $g_{i,2(b-1)}$ contains the sink~$t_A$
  of~$G_A$ (see also
  \cref{fig:reduction-whole}).
  Consider a path~$\pi_j$ with $1 \leq j \leq b$, and let
  $\{a_{\chi},a_{\lambda},a_{\mu}\}$ be the $j$-th bin of the solution
  of $I$.  Let $\nu$ be an index such that $1 \leq \nu \leq b$.
\begin{figure}[tbp]
	\begin{subfigure}[b]{0.60\textwidth}
		\centering
            \includegraphics[page=5]{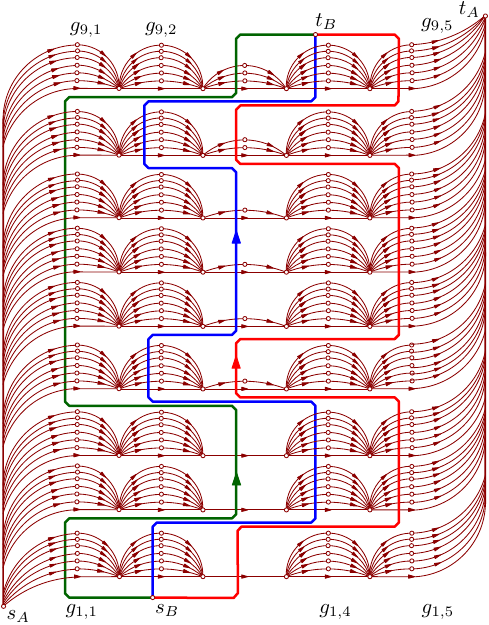}\hfill
		\subcaption{}
        \label{fig:reduction-whole-a}
	\end{subfigure}
        \hfill
	\begin{subfigure}[b]{0.30\textwidth}
		\centering
            \includegraphics[page=6]{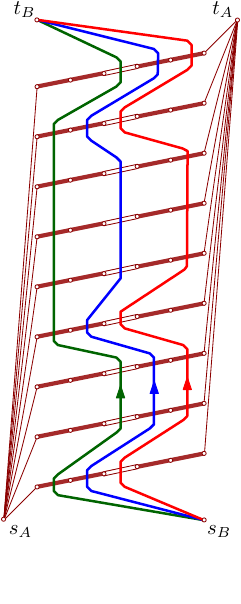}
		\subcaption{}
		\label{fig:reduction-whole-b}
	\end{subfigure} 
        \caption{(a) Digraph $G_A$ (dark red) and a schematic
          representation of digraph $G_B$ where each colored curve
          represents a directed path with $W+(k-3)(W+1)$ edges. The
          corresponding instance of \textsc{3-Partition} is
          $I = \{1,1,1,2,2,2,2,3,4\}$, with $b=3$ and $W=6$. The
          1-planar drawing corresponds to the solution $\{1,1,4\}$
          (green path), $\{2,2,2\}$ (blue path), and $\{1,2,3\}$ (red
          path). The drawing in (a) is not upward but it can be made
          upward by stretching it vertically as~shown~in~(b), where
          thick edges represent $(q)$-gates and the central
          white-filled edges represent $(a)$-gates.
        }
   \label{fig:reduction-whole}
\end{figure}

  \begin{itemize}
  \item If $\nu$ is one of $\chi, \lambda, \mu$, then $\pi_j$ crosses
    the $(a_\nu)$-gate of the $(b-1,W+1,a_\nu)$-chain (see, for example,
    the red path crossing $g_{1,3}$ in \cref{fig:reduction-whole}).
  \item If $\nu$ is not one of $\chi, \lambda, \mu$ and the path $\pi_h$
    that crosses the $(a_{\nu})$-gate of the $(b-1,W+1,a_\nu)$-chain is
    such that $h > j$, then $\pi_j$ crosses the gate $g_{\nu,j}$ of the
    $(W+1,b-1,a_\nu)$-chain (see, for example, the blue path crossing
    $g_{2,2}$ in \cref{fig:reduction-whole}).
  \item Otherwise ($h < j$), path $\pi_j$ crosses the gate
    $g_{\nu,b-1+j}$ of the $(b-1,W+1,a_\nu)$-chain (see, for example, the red
    path crossing $g_{5,5}$ in
    \cref{fig:reduction-whole}).
  \end{itemize}
  By the procedure above, if $\pi_j$ crosses a gate $g(\nu,q)$, with
  $1 \leq \nu \leq k$ and $1 \leq q \leq 2(b-1)$, there is no other
  $\pi_h$, with $h \neq j$, such that $\pi_h$ crosses $g(\nu,q)$.
  Also, the number of edges that $\pi_j$ crosses is
  $a_\chi + a_\lambda + a_\mu + (k-3)(W+1) = W+(k-3)(W+1)$, which is
  the number of edges of $\pi_j$. Hence, it is possible to draw
  $\pi_j$ in such a way that each of its edges is crossed exactly once
  and no edge of $G_A$ is crossed more than once.  This concludes the
  proof.
\end{proof}

\begin{theorem}	
\label{th:st-negative}
    Testing upward 1-planarity is \NP-complete even 
    in the following restricted cases:
    \begin{enumerate}
        \item\label{ca:fixed-sp-1s-1t} The bimodal rotation system is fixed, the DAG has exactly one source and exactly one sink, the underlying graph is series-parallel.
        \item\label{ca:variable-rigid-1s-1t} The bimodal rotation system is not fixed, the DAG has exactly one source and exactly one sink, the underlying planar graph is obtained by replacing the edges of a $K_4$ with series-parallel graphs.
        \item\label{ca:variable-sp-3s-2t} The bimodal rotation system is not fixed, the underlying graph is series-parallel, there is one source and two sinks.   
    \end{enumerate}
\end{theorem}

\begin{proof}
  It is immediate to observe that upward 1-planarity testing is in the
  NP class of complexity (one can guess an upward embedding and test
  it in polynomial time). We show that it is NP-hard.  For each case
  in the statement it suffices to exhibit a digraph $G$ that contains
  $G_A$ and $G_B$ as sub-graphs and that satisfies the conditions
  of~\cref{le:3-partition}.

  Let $m_A$ and $m_B$ be the number of edges of $G_A$ and $G_B$,
  respectively.  We define a \emph{barrier} to be an
  st-digraph consisting of a $(d,2)$-parallel, where $d = m_A + m_B + 1$. Note
  that no $s_At_A$-path ($s_Bt_B$-path) can fully cross a barrier in
  such a way that every edge of the path is crossed at most once. This
  implies that neither $G_A$ nor $G_B$ can fully cross a barrier.
\begin{figure}[ht]
	\begin{subfigure}[b]{0.13\textwidth}
		\centering
		\includegraphics[page=1]{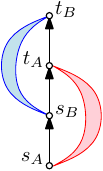}
		\subcaption{}
            \label{fig:ca:fixed-sp-1s-1t-a}
	\end{subfigure}
        \hfill
	\begin{subfigure}[b]{0.10\textwidth}
		\centering
		\includegraphics[page=2]{1-planarity-hardness.pdf}
		\subcaption{}
		\label{fig:ca:fixed-sp-1s-1t-b}
	\end{subfigure} 
        \hfill
	\begin{subfigure}[b]{0.16\textwidth}
		\centering
		\includegraphics[page=5]{1-planarity-hardness.pdf}
		\subcaption{}
		\label{fig:ca:variable-rigid-1s-1t-a}
	\end{subfigure}
         \hfill
	\begin{subfigure}[b]{0.16\textwidth}
		\centering
		\includegraphics[page=6]{1-planarity-hardness.pdf}
		\subcaption{}
		\label{fig:ca:variable-rigid-1s-1t-b}
	\end{subfigure} 
        \hfill
	\begin{subfigure}[b]{0.16\textwidth}
		\centering
		\includegraphics[page=7]{1-planarity-hardness.pdf}
		\subcaption{}
		\label{fig:ca:variable-sp-1s-2t-a}
	\end{subfigure} 
        \hfill
	\begin{subfigure}[b]{0.16\textwidth}
		\centering
		\includegraphics[page=8]{1-planarity-hardness.pdf}
		\subcaption{}
		\label{fig:ca:variable-sp-1s-2t-b}
	\end{subfigure} 
        \hfill
	\caption{Some digraphs for the proof of \cref{th:st-negative}. Thick black edges represent barriers.}
	\label{fig:1-planarity-hardness}
\end{figure}

\begin{itemize}
  \item Case~\ref{ca:fixed-sp-1s-1t}: Refer to
    \cref{fig:ca:fixed-sp-1s-1t-a},
    where the thick edges schematically
    represent barriers
    $B_1$ (from $s_A$ to $s_B$), $B_2$ (from $s_B$ to $t_A$), and
    $B_3$ (from $t_A$ to $t_B$); in the figure, the red shaded shape
    represents $G_A$ and the blue shaded shape represents $G_B$. The
    figure schematically represents a series-parallel digraph $G$ with
    source $s_A$ and sink $t_B$. As described in
    \cref{fig:ca:fixed-sp-1s-1t-b}, an upward 1-planar drawing where the edges of $G_A$ precede the
    edges of $B_1$ in the left to right order of the edges exiting
    $s_A$ and the edges of $G_B$ precede the edges of $B_3$ in the
    left to right order of the edges entering $t_B$ exists if and only
    if $G_A$ crosses $G_B$. In fact, in any upward 1-planar drawing of
    $G$ that preserves its bimodal rotation system $G_A$ crosses
    $G_B$; conversely, starting from an upward 1-planar drawing of the
    union of $G_A$ and $G_B$ in which $G_A$ and $G_B$ cross, an upward
    1-planar drawing of $G$ that preserves the bimodal rotation system
    of \cref{fig:ca:fixed-sp-1s-1t-a} can be obtained as shown in
    \cref{fig:ca:fixed-sp-1s-1t-b}.
  \item Case~\ref{ca:variable-rigid-1s-1t}: We modify a very well
    known small digraph that is traditionally used to exhibit a
    digraph that is planar but not upward
    planar~\citep*{bertolazzi-siam,dett-gd-99}. As above, in
    \cref{fig:ca:variable-rigid-1s-1t-a} thick edges schematically
    represent barriers, the red shaded shape represents $G_A$, and the
    blue shaded shape represents $G_B$. Again, an upward 1-planar
    drawing of this digraph exists if and only if $G_A$ fully crosses
    $G_B$ (see \cref{fig:ca:variable-rigid-1s-1t-b}). In particular,
    neither $s_B$ nor $t_B$ can be drawn in the finite region of the
    plane bounded by $G_A$ since at least one $s_At_A$-path
    would fully cross a barrier. Similarly, neither $s_A$ nor $t_A$
    can be drawn in the finite region of the plane bounded by
    $G_B$. Observe that the digraph has a single source and a single
    sink and its underlying planar graph is obtained by replacing the
    edges of a $K_4$ with series-parallel
    graphs.%
  \item Case~\ref{ca:variable-sp-3s-2t}: Consider the digraph $G$ of
    \cref{fig:ca:variable-sp-1s-2t-a}, where the thick edges
    schematically represent barriers; as in the previous case, the red
    shaded shape represents $G_A$ and the blue shaded shape represents
    $G_B$. With analogous arguments as in the previous case, it is
    immediate to see that $G$ has an upward 1-planar drawing if and
    only if $G_A$ fully crosses~$G_B$ (see
    \cref{fig:ca:variable-sp-1s-2t-b}).
  \end{itemize}
\end{proof}

The following corollary is a consequence of the argument used to prove the second case in the statement of~\cref{th:st-negative}.

\begin{corollary}
  \label{co:fixed-embedding-1-planarity}
  Testing upward 1-planarity is \NP-complete for single source-single
  sink DAGs with a fixed bimodal rotation system, whose underlying
  planar graph is obtained by replacing the edges of a $K_4$ with
  series-parallel graphs.
\end{corollary}

We conclude this section by remarking some differences between the complexity of upward planarity testing and upward 1-planarity testing. 
When the bimodal rotation system is fixed, upward planarity testing can be solved in polynomial time~\citep{DBLP:journals/algorithmica/BertolazziBLM94}, whereas upward 1-planarity testing is \NP-hard (\cref{th:st-negative}). Also, when the bimodal rotation system is not fixed and the digraph has a constant number of sources and sinks, differently from upward 1-planarity testing, upward planarity testing can again be solved in polynomial time~\citep*{ChaplickGFGRS22}. On the other hand, any digraph whose bimodal rotation system is not fixed, whose underlying graph is series-parallel, and that has only one source and only one sink is always upward planar and thus upward 1-planar. Indeed, adding an edge between any two vertices of the undirected underlying series-parallel graph yields a planar graph
\citep[see, e.g.,][]{dt-olpt-96}. It follows that $G$ can be turned into a planar st-graph by connecting its source to its sink by an edge and hence it is upward planar~\citep{gt-upt-95}. This discussion is summarized in \cref{tab:complexity-table}.

\begin{table}[htb]
    \caption{A comparison between results in the literature about the
      complexity of testing upward planarity and the results discussed
      in this paper about the complexity of testing upward
      1-planarity.}
    \label{tab:complexity-table}

    \smallskip
    
    \centering\footnotesize
    \begin{tabular}{|@{}c@{}|@{}c@{}||@{}c@{}|@{}c@{}|@{}c@{}|@{}c@{}|}
    \hline
    \bf Underlying\rule{0ex}{2.3ex} & 
    \bf Acyclic & 
    \multicolumn{2}{c|}{\bf Upward planarity} &
    \multicolumn{2}{c|}{\bf Upward 1-planarity} \\
    \cline{3-6}
     
    \bf planar\rule{0ex}{2.3ex} & \bf orientation: & \multicolumn{2}{c|}{Embedding} & \multicolumn{2}{@{}c@{}|}{Rotation system} \\
     
    \bf graph\rule{0ex}{2.3ex} & \bf sources/sinks & fixed & variable & fixed & variable \\
    &&\multicolumn{1}{c}{}&&&\\[-1.9ex]
    \hline\hline
     
    \multirow{2}{*}{%
    \parbox[c][9ex][c]{1.7cm}{
    \centering\bf
      Series- parallel}}
    &
    \bf multi/multi &
    \multicolumn{2}{c|}{polynomial \citep{upward-spirality,cdf-tup-gd-2022}} 
    &
    \multirow{2}{*}{%
    \parbox[c][9ex][c]{1.9cm}{%
             \centering \NP-complete \\ \cref{th:st-negative} 
            \\ Case~\ref{ca:fixed-sp-1s-1t}} 
    }
    &
    \parbox[c][10ex][c]{1.9cm}{%
        \centering \NP-complete \\ 
        \cref{th:st-negative} \\ 
        Case~\ref{ca:variable-sp-3s-2t}
    }  
    \\ 
    \cline{2-4}\cline{6-6}
    
    & 
    \bf single/single
    &
    \multicolumn{2}{c|}{polynomial~\citep{ChaplickGFGRS22,cdf-tup-gd-2022,upward-spirality}}
    &
    & 
    \parbox[c][7ex][c]{2cm}{%
        \centering trivially \mbox{polynomial}
    }
    \\
    \hline
    
    \multirow{2}{*}{%
    \parbox[c][7ex][c]{1.7cm}{
    \centering\bf General}
    }
    &
    \bf multi/multi
    &
    \parbox[c][7ex][c]{3.5cm}{\centering polynomial \\\citep{DBLP:journals/algorithmica/BertolazziBLM94}}
    &
    \parbox[c][7ex][c]{3.5cm}{\centering \NP-complete \mbox{\citep{DBLP:journals/siamcomp/GargT01}}}
    &
    \multirow{2}{*}{%
    \parbox[c][7ex][c]{1.9cm}{
            \centering \NP-complete \\ \cref{co:fixed-embedding-1-planarity}}  
    }
    &      
    \multirow{3}{*}{%
        \parbox[c][7ex][c]{1.9cm}{
            \centering \NP-complete \\ 
            \cref{th:st-negative} \\ 
            Case~\ref{ca:variable-rigid-1s-1t}
        } 
    } \\   
    \cline{2-4}
     
    & 
    \parbox[c][5ex][c]{2cm}{\centering\bf single/single}
    &
    \multicolumn{2}{c|}{%
        polynomial~\citep{ChaplickGFGRS22}
    } & & \\  
    \hline
    \end{tabular}
\end{table}

\section{Testing Upward Outer-1-Planarity}\label{sec:outerplanar}

\Cref{th:st-negative} motivates the study of the complexity of testing upward 1-planarity with additional constraints. A common restriction in the study of beyond-planar graph drawing problems is the one that requires that all vertices are incident to the same face.
Specifically, a graph is \emph{upward outer-1-planar} if it admits an upward 1-planar embedding in which every vertex lies on the outer face. 
This section is devoted to the proof of the following result.

\begin{theorem}
	\label{th:outer-testing-short}
	For single-source DAGs, upward outer-1-planarity can be tested in linear time.
\end{theorem}

\subsection{Basic Facts and Definitions}

We provide the following characterization for the single-source DAGs admitting an upward outer-1-planar drawing.

\begin{theorem}\label{th:characterization}
	A single-source graph is upward outer-1-planar if and only if it admits an outer-1-planar embedding whose planarization is acyclic.
\end{theorem}

\begin{proof}
	Let~$G$ be a single-source graph with an upward outer-1-planar embedding~$\Gamma$.  Clearly, $\Gamma$ is outer-1-planar.  Moreover, planarizing~$\Gamma$ yields an upward drawing of the planarization, hence the planarization is acyclic.
	
	Conversely, let~$\Gamma$ be an outer-1-planar embedding whose planarization~$\Gamma^\star$ is acyclic.  Observe that no crossing vertex of~$\Gamma^\star$ is a source or a sink.  Therefore,~$\Gamma^\star$ only has a single source~$s$ and the sinks of $G$, which are all incident to the outer face. Let~$\Gamma^+$ be obtained from~$\Gamma^\star$ by adding a new super sink~$t$, with edges from all sinks of $G$ to~$t$ into the outer face.  Clearly~$\Gamma^+$ is acyclic and a planar $st$-graph, and therefore upward planar~\citep{batTam:88}.  Removing~$t$ then yields an upward planar drawing of~$\Gamma^\star$ and hence~$\Gamma$ is upward outer-1-planar.
\end{proof}

The following lemma follows from \cref{th:characterization} and shows that we can decompose the problem on the biconnected components of the input graph.

\begin{lemma}
	\label{lem:o1p-bico}
	A single-source graph is upward outer-1-planar if and only if each of its blocks admits an upward outer-1-planar embedding.
\end{lemma}

\begin{proof}
	The necessity is trivial.
	For the sufficiency, suppose that all blocks admit an upward outer-1-planar embedding, whose planarization is thus acyclic. Combine all such planarizations at the cut-vertices so that all their vertices are incident to the outerface. Since each cycle in the graph is contained inside a block, this yields an acyclic planarization of an outer-1-planar embedding of the entire graph. This and \cref{th:characterization} imply the statement.
\end{proof}

The SPQR-tree data structure, introduced by \cite{dt-olpt-96}, to
represent the decomposition of a biconnected graph into its
triconnected components is a special type of decomposition
tree~\citep{DBLP:journals/siamcomp/HopcroftT73}.
A \emph{decomposition tree} of a biconnected graph $G$ is a
tree~$\mathcal T$ whose nodes $\mu$ are associated with a biconnected
multi-graph $\skel(\mu)$, called \emph{skeleton}. The edges of a
skeleton can be either \emph{real} or \emph{virtual}, and for each
node~$\mu$, the virtual edges of~$\skel(\mu)$ correspond bijectively
to the edges of~$\mathcal T$ incident to~$\mu$. The tree $\mathcal T$
can be inductively defined as follows.  In the base case, $\mathcal T$
consists of a single node $\mu$ whose skeleton is the graph $G$
consisting solely of real edges.  In the inductive case, let $\mu$ be
any node of $\mathcal T$ and let $H_1$ and~$H_2$ be edge-disjoint
connected subgraphs of~$\skel(\mu)$ such that
$\skel(\mu) = H_1 \cup H_2$ and $H_1 \cap H_2 = \{u, v\}$. We obtain a
new decomposition tree $\mathcal T'$ from $\mathcal T$ by splitting
the node $\mu$ into two adjacent nodes $\nu_1$ and
$\nu_2$ whose skeletons are $H_1 + uv$ and $H_2 + uv$,
respectively, where~$uv$ is a new virtual edge that corresponds to the
edge~$\nu_1\nu_2$ of~$\mathcal T'$; also, we replace each edge $\mu\tau$ of $\mathcal T$ either with the edge $\nu_1 \tau$ or with the edge $\nu_2 \tau$ depending on wether edge $\mu\tau$ corresponds to a virtual edge of $H$ that belongs to $H_1$ or to $H_2$, respectively.
The edges of $\mathcal T$ incident to~$\mu$ are distributed among~$\nu_1$ and~$\nu_2$ based on whether their corresponding virtual edge belongs to~$H_1$ or~$H_2$.
Let now~$\mathcal T$ be an arbitrary decomposition tree of a
biconnected graph.  Note that each edge of $G$ occurs as a real edge
in precisely one skeleton of $\mathcal T$. Consider a virtual edge
$e$ in the skeleton of some node $\mu$ of $\mathcal T$ and let
$\varepsilon$ be the edge of $\mathcal T$ which corresponds to $e$. The \emph{expansion graph} $\expn(e)$ of $e$ is
the graph obtained as the union of the real edges belonging to the
skeletons of the subtree of $\mathcal T$ reachable from $\mu$ via the
edge $\varepsilon$.
Observe that $\expn(\{a,b\})$ is connected to the rest of the graph
through $a$ and $b$. The vertices~$a$ and $b$ are the \emph{poles}
of~$\expn(\{a,b\})$.  For every pair of adjacent nodes $\mu$ and $\nu$
of $\mathcal{T}$, there exists a virtual edge $e$ in $\skel(\mu)$ and
a virtual edge $e'$ in $\skel(\nu)$ with the same end-vertices. We say
that $\mu$ (resp.\ $\nu$) is the \emph{refining} node $\refn(e')$ of
$e'$ (resp.\ $\refn(e)$ of $e$).

The SPQR-tree~$\mathcal T$ of a biconnected graph~$G$ is a decomposition tree that has three types of nodes.  The skeleton of an \emph{S-node} is a simple cycle of length at least~3, the skeleton of a \emph{P-node} consists of two vertices and at least three parallel edges, and the skeleton of an \emph{R-node} is a simple 3-connected graph.  Moreover, no two S-nodes and no two P-nodes are adjacent in $\mathcal T$.  We remark that unlike the classical definition of SPQR-trees, we allow real edges in any skeleton; this avoids the need to use Q-nodes whose skeletons have two vertices and one real and one virtual edge.

The SPQR-tree of a biconnected planar graph $G$ can be used to succinctly represent all planar embeddings of $G$.  Specifically, any planar embedding of $G$ uniquely defines an embedding for all the skeletons of $\mathcal T$. Moreover, recursively combining planar embeddings of the skeletons of $\mathcal T$ via 2-clique sums of such embeddings results in a planar embedding of $G$.
The \emph{2-clique-sum} of two embeddings $\mathcal E'$ and $\mathcal E''$ containing the virtual edge $\{u,v\}$ identifies the two copies of $u$ and the two copies of $v$, removing the edge $\{u,v\}$, in such a way that edges in $\mathcal E'$ do not alternate with edges in $\mathcal E''$. Even in the presence of non-planar embeddings of skeletons, such an operation can be performed in the same fashion as long as the common virtual edge $\{u,v\}$ is crossing free in both embeddings.

Let~$\mu$ be a P-node with skeleton vertices~$\{u,v\}$, let~$e$ be a
virtual edge in~$\skel(\mu)$ that is refined by an S-node~$\lambda$
with skeleton~$(u,c_1,c_2,\dots,c_k,v,u)$.  Then the \emph{first segment}
of~$e$ at~$u$ is the (virtual or real) edge~$\{u,c_1\}$ of~$\skel(\lambda)$ and the
\emph{first segment} of~$e$ at~$v$ is the (virtual or real)
edge~$\{c_k,v\}$ of~$\skel(\lambda)$.

\subsection{Proof of \cref{th:outer-testing-short}}

There exist two different algorithms for testing outer-1-planarity, namely by \citet*{DBLP:conf/gd/HongEKLSS13,DBLP:journals/algorithmica/HongEKLSS15}
and by \citet*{DBLP:journals/algorithmica/AuerBBGHNR16}.  It is likely that both can be adapted to test upward outer-1-planarity.  We choose to follow the approach by \cite{DBLP:journals/algorithmica/AuerBBGHNR16}, which we briefly summarize in the following, as it seems slightly simpler to extend to our setting.  Two key properties that both papers leverage are the facts that if~$G$ is an outer-1-planar graph and~$\mathcal T$ is its SPQR-tree, then the skeleton of each R-node of~$\mathcal T$ is a $K_4$ that contains two non-adjacent real edges and the skeleton of each P-node has at most four virtual edges and at most one real edge.  In what follows, we assume that these two properties are satisfied.

\subsubsection{The Outer-1-Planarity Testing Algorithm of Auer et al.}

By \Cref{lem:o1p-bico}, we may assume that our input graph $G$ is biconnected and acyclic.  
Consider a decomposition tree~$\mathcal T$ of $G$ (not necessarily the SPQR-tree).  Let~$\mu$ be a node of~$\mathcal T$.  We call an outer-1-planar embedding of~$\skel(\mu)$ \emph{good} if only real edges cross and each virtual edge is incident to the outer face. 
\begin{observation}\label{obs:good-embeddings}
	Good embeddings of all skeletons of~$\mathcal T$ together define a unique outer-1-planar embedding of~$G$, which is obtained by forming the 2-clique sums of the embeddings of the skeletons.
\end{observation}

However, if we consider the SPQR-tree, good embeddings of the skeletons may not exist even if~$G$ is outer-1-planar.  Therefore, \cite{DBLP:journals/algorithmica/AuerBBGHNR16} work with a more general definition of outer-1-planar embeddings for skeletons of the SPQR-tree and show that an outer-1-planar embedding of $G$ can be combinatorially described by embeddings of the skeletons of its SPQR-tree~$\mathcal T$ that satisfy the following five conditions {C1}--{C5}.  We note that the edges of the skeletons may cross even for P-nodes, i.e., the embeddings of the skeletons we consider are not necessarily simple drawings.

\begin{enumerate}[C1]
\item \namedlabel{c1}{C1} Each skeleton is embedded outer-1-planar
  such that each virtual edge has a segment that is incident to the
  outer face \citep[Proposition
  3]{DBLP:journals/algorithmica/AuerBBGHNR16}.
    
\item \namedlabel{c2}{C2} Skeletons of S-nodes are embedded planarly
  \citep[Corollary 3]{DBLP:journals/algorithmica/AuerBBGHNR16}.
    
\item \namedlabel{c3}{C3} Skeletons of R-nodes are embedded such that virtual edges are crossing-free and
  precisely two non-virtual edges cross each other \citep[Corollary
  2]{DBLP:journals/algorithmica/AuerBBGHNR16}.

\item \namedlabel{c4}{C4} If an edge~$e=\{u,v\}$ of the skeleton of a
  P-node is crossed, then it is a virtual edge corresponding to an
  S-node.  Moreover, if the segment of~$e$ incident to~$u$ (incident
  to~$v$) is not incident to the outer face, then the first segment
  of~$e$ at $u$ (at $v$) must be a real edge \citep[Lemma
  2]{DBLP:journals/algorithmica/AuerBBGHNR16}.
\item \namedlabel{c5}{C5} Let~$\mu$ be an S-node whose skeleton
  contains a real edge~$\{u,v\}$ and two adjacent virtual
  edges~$\{u',u\}$ and~$\{v,v'\}$ that are refined by~$\rho$ and
  $\lambda$, respectively.  Let~$e_\rho,e_\lambda$ be the virtual
  edges that represent~$\mu$ in~$\skel(\rho)$ and~$\skel(\lambda)$,
  respectively.  If both~$\rho$ and~$\lambda$ are P-nodes, then the
  segment of~$e_\rho$ incident to~$u$ must be incident to the outer
  face in the embedding of~$\skel(\rho)$ or the segment of~$e_\lambda$
  incident to~$v$ must be incident to the outer face in the embedding
  of~$\skel(\lambda)$.
\end{enumerate}
      
\begin{figure}[tb!]
  \centering
  \includegraphics{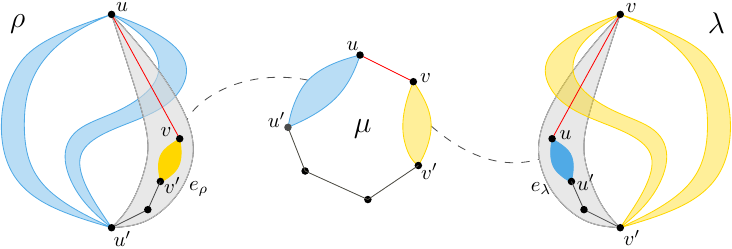}
  \caption{A portion of an SPQR-tree with an S-node~$\mu$ and two adjacent S-nodes~$\rho,\lambda$ illustrating the necessity of Condition~\ref{c5}. Since
    neither the segment of~$e_\rho$ incident to~$u$ is incident to the
    outer face in the embedding of~$\skel(\rho)$ nor the segment
    of~$e_\lambda$ incident to~$v$ is incident to the outer face in
    the embedding of~$\skel(\lambda)$, the real edge $\{u,v\}$ must be
    involved in at least two crossings. Namely, at least one crossing
    with an edge in the expansion graph of the blue virtual edges of
    $\skel(\rho)$ and at least one crossing with an edge in the
    expansion graph of the yellow virtual edges of $\skel(\lambda)$.}
  \label{fig:C4-necessity}
\end{figure}

The idea behind Conditions~\ref{c4} and \ref{c5} is that crossing two virtual edges~$e,e'$ in a skeleton of a P-node~$\mu$ with vertices~$\{u,v\}$ in such a way that~$e$ is not incident to the outer face at~$u$ and~$e'$ is not incident to the outer face at~$v$ implies that we cross the first segment of~$e$ at~$u$ and the first segment of~$e'$ at $v$; see the gray-shaded regions in \cref{fig:C4-necessity} and the more formal discussion later.  Condition~\ref{c4} guarantees that these are real edges.  On the other hand, Condition~\ref{c5} ensures that each real edge receives at most one crossing in this way; see \cref{fig:C4-necessity} for an illustration of a case in which this is violated.

It follows from the work of \cite{DBLP:journals/algorithmica/AuerBBGHNR16}
that together these conditions describe all outer-1-planar embeddings of $G$ without unnecessary crossings (i.e., there is no outer-1-planar drawing whose crossing edge pairs form a strict subset).  We summarize this in the following theorem.  As this result does not appear explicitly in the work of Auer et al., for the sake of completeness, we briefly sketch the construction. In particular, the construction of an embedding of $G$ from the embeddings of the skeletons is crucial for understanding our next steps. 

\begin{theorem}
	\label{thm:spqr-outer-1-planar}
	Let~$G$ be a biconnected graph and let~$\mathcal T$ be its SPQR-tree.  There is a bijection between the outer-1-planar embeddings of~$G$ without unnecessary crossings and the choice of an embedding for each skeleton of~$\mathcal T$ that satisfy Conditions~\ref{c1}--\ref{c5}.  Moreover, both directions of the bijection can be computed in linear time.    
\end{theorem}

\begin{proof}
    For the bijection, observe that the necessity of \ref{c1}--\ref{c5} has been argued above.  For the converse, we show how we obtain an embedding of~$G$ from embeddings~$\{\mathcal E_\mu\}_{\mu \in V(\mathcal T)}$ for each skeleton $\skel(\mu)$ that satisfy~\ref{c1}--\ref{c5}. We refer to this embedding as the \emph{combined embedding of $\{\mathcal E_\mu\}_{\mu \in V(\mathcal T)}$}.

    \begin{figure}[tb!]
	   \centering
	   \includegraphics[page=2]{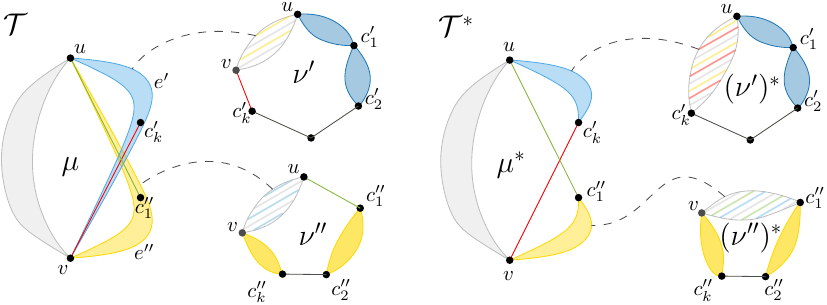}
	   \caption{Illustration for the extension of a P-node $\mu$.}
	   \label{fig:extension}
    \end{figure}

First note that, if there are no crossings that involve virtual edges
in the embeddings of the skeletons, then the embeddings are good and
their 2-clique sum yields an outer-1-planar embedding of~$G$ by
\cref{obs:good-embeddings}.  If there are virtual edges that are
involved in a crossing, they belong to the skeletons of P-nodes by
Conditions~\ref{c2} and~\ref{c3}.  We reduce to the case of good
embeddings by \emph{extending} such P-nodes, that is, by modifying the
SPQR-tree, its skeletons, and their embeddings as follows; refer to
\cref{fig:extension}.  Let~$e',e''$ be two crossing virtual edges in
an embedding of a skeleton of a P-node~$\mu$. Let~$\{u,v\}$ be the
vertices of~$\skel(\mu)$.  By Condition~\ref{c4}, we have that the
nodes $\nu'= \refn(e')$ and $\nu''= \refn(e'')$ are S-nodes with
skeletons~$(u,c'_1,\dots,c'_k,v,u)$ and~$(u,c_1'',\dots,c_{\ell}'',v,u)$,
respectively.  Without loss of generality, we may assume that the
segment of~$e'$ incident to~$u$ is incident to the outer face in the
embedding of~$\skel(\mu)$, and hence by Condition~\ref{c1} the segment
of~$e''$ incident to~$v$ is also incident to the outer face. By
Condition~\ref{c4}, the edge~$\{c_k',v\}$
of~$\skel(\nu')$, which is the first
segment of~$e'$ at~$v$, and the edge~$\{u,c_1''\}$ of~$\skel(\nu'')$,
which is the first segment of~$e''$ at~$u$, are real edges.  We perform
the following modifications.  First, we shorten the skeletons
of~$\nu'$ and of~$\nu''$ to~$(u,c'_1,\dots,c'_{k-1},c'_k,u)$
and~$(c''_1,c''_2,\dots,c''_{\ell},v,c''_1)$, respectively, and keep
their embeddings planar.  Second, we subdivide the virtual edge $e'$
($e''$) with the vertex $c'_k$ ($c''_1$) such that $\{v,c'_k\}$
($\{u,c''_1\}$) is a real edge and $\{c'_k,u\}$ ($\{c''_1,v\}$) is a
virtual edge that is refined by the modified S-node $\nu'$
($\nu''$). We embed the skeletons of such nodes so that the real
edges~$\{v,c'_k\}$ and $\{u,c''_1\}$ cross and the two virtual
edges~$\{u,c'_k\}$ and $\{v,c''_1\}$ are incident to the outer
face. Note that this essentially {\em moves} two real edges from
S-node skeletons into an adjacent P-node. See the right side of
\cref{fig:extension} for the modified skeletons.

Condition~\ref{c5} guarantees that no real edge needs to be moved into the skeletons of two different P-nodes.  Hence we can apply this operation simultaneously and independently to all pairs of crossing virtual edges.  We then arrive at a decomposition tree~$\mathcal T^*$ of~$G$, whose skeletons have good embeddings.  As explained above, we then obtain an outer-1-planar embedding of $G$.

Observe that this construction can be carried out in linear time.  Conversely, given an outer-1-planar embedding~$\mathcal E$ of $G$.  The embeddings~$\mathcal E_\mu$ of all skeletons~$\skel(\mu)$ can be obtained by checking the order of the incident edges for each vertex and the pairs of edges that cross.  Moreover, it is then clear that applying the above construction to the~$\mathcal E_\mu$, we reobtain~$\mathcal E$.
\end{proof}

The tree~$\mathcal T^\star$ we constructed towards the end of the proof of \Cref{thm:spqr-outer-1-planar} is called the \emph{extension of~$\mathcal T$ with respect to $\{\mathcal E_\mu\}_{\mu \in V(\mathcal T)}$}.
Since the modification neither adds nor removes nodes of~$\mathcal T$, there is a bijection between the nodes of~$\mathcal T$ and~$\mathcal T^*$. Namely, each node~$\mu$ of~$\mathcal T$ contributes a unique node~$\mu^*$ to~$\mathcal T^*$, and these are all the nodes of~$\mathcal T^*$.  
The node~$\mu^*$ is called the \emph{extended node of~$\mu$}
and the good embedding of~$\skel(\mu^*)$ obtained with the procedure described above is called the \emph{extended embedding} of~$\skel(\mu)$.  We further note that for P- and R-nodes~$\mu$,~$\skel(\mu^\star)$ only depends on the embedding~$\mathcal E_\mu$ of~$\skel(\mu)$.  Thus, for P-nodes and R-nodes~$\skel(\mu^\star)$ can also be defined if we only have an embedding of~$\skel(\mu)$ that satisfies~\ref{c1}--\ref{c4}.  Condition \ref{c5} is only needed if we want to combine such skeletons into the whole extension~$\mathcal T^\star$.

\subsubsection{An Algorithm to test Upward Outer-1-Planarity.}
Our goal is to find an outer-1-planar embedding of $G$ whose planarization is acyclic.  Observe that it suffices to consider outer-1-planar embeddings without unnecessary crossings, since adding more crossings cannot make a cyclic planarization acyclic.  Our goal is therefore to find outer-1-planar embeddings of the skeletons of the SPQR-tree $\mathcal{T}$ of $G$ satisfying Conditions~\ref{c1}--\ref{c5} so that the resulting planarization of~$G$ is acyclic.

\paragraph{Directing virtual edges.}

To keep track of the existence of directed paths, we orient some edges
of the skeletons of~$\mathcal T$.  Consider a skeleton~$\skel(\mu)$
and let~$e=\{u,v\}$ be an edge of~$\skel(\mu)$.  We orient~$e$
from~$u$ to~$v$ if either it is a real edge directed from~$u$ to~$v$
in~$G$ or if~$\expn(e)$ contains a directed path from~$u$ to~$v$.

We observe some properties of the constructed orientation.
First, note that a virtual edge~$\{u,v\}$ is oriented in at most one
direction, but it may also remain undirected in case its expansion
graph contains neither a directed path from~$u$ to~$v$ nor a directed
path from~$v$ to~$u$.  Second, if a virtual edge~$e$ with endpoints $u,v$ is
directed,
say from~$u$ to~$v$, then~$u$ is the single source of~$\expn(e)$, as otherwise $G$ would contain an additional source.  Finally, if $e$ is
undirected, since~$G$ has a single source $s$, either~$\expn(e)$
contains $s$ as a non-pole vertex or both~$u$ and~$v$ are sources
of~$\expn(e)$.

\paragraph{Acyclic Embeddings.} 

\begin{figure}[tb]
	\begin{minipage}[b]{.48\textwidth}
		\centering
		\includegraphics[page=3]{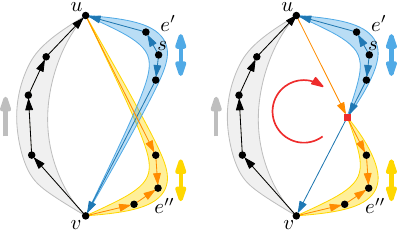}
		\subcaption{}
		\label{fig:acyclic-1-long}
	\end{minipage}
	\hfill
	\begin{minipage}[b]{.48\textwidth}
		\includegraphics[page=4]{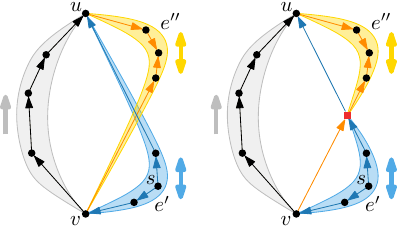}
		\subcaption{}
		\label{fig:acyclic-2-long}
	\end{minipage}
	\caption{Two embeddings of the skeleton of a P-node $\mu$ and the corresponding planarizations. The planarization in (a) contains a directed cycle, the one in (b) does not. Thick arrowed edges show the direction of the virtual edges; a double arrow indicates an undirected virtual edge.}
	\label{fig:cyclic-long}
\end{figure}

Suppose that the edges of the skeletons are directed as described
above, let~$\mu$ be a node with an embedding~$\mathcal E_\mu$ that
satisfies~{C1}--{C4}.
Recall from the remark after the proof of \cref{thm:spqr-outer-1-planar} that~$\skel(\mu^\star)$ and its embedding~$\mathcal E_\mu^\star$ depend only on~$\mathcal E_\mu$.  We call~$\mathcal E_\mu$ \emph{acyclic} if the planarization of the extended embedding~$\mathcal E_\mu^\star$ is
acyclic; see \cref{fig:cyclic-long} for examples.  It is not
surprising that indeed such embeddings are necessary for the existence
of an upward outer-1-planar embedding of~$G$.

\begin{lemma}[Necessity]
  \label{le:upwardG-implies-acyclic-skeletons}
  Let~$\mathcal E$ be an upward outer-1-planar embedding of $G$
  without unnecessary crossings.  Then the embedding of~$\skel(\mu)$
  induced by~$\mathcal E$ is acyclic for each P- or R-node~$\mu$ of the
  SPQR-tree~$\mathcal T$ of $G$.
\end{lemma}

\begin{proof}
  By \cref{thm:spqr-outer-1-planar},
  $\mathcal E$ induces embeddings~$\mathcal E_\mu$ for each
  skeleton~$\skel(\mu)$ of~$\mathcal T$, which satisfy the
  Conditions~\ref{c1}--\ref{c5}.  Let~$\mathcal T^*$ be the extension
  of~$\mathcal T$ with respect to these embeddings and recall from the proof of \Cref{thm:spqr-outer-1-planar} that~$\mathcal E$ can be obtained as the 2-clique sum of the embeddings~$\mathcal E_\mu^\star$.  Let~$\mu$ be a
  P- or R-node of~$\mathcal T$ with embedding~$\mathcal E_\mu$ and let~$\mu^*$
  be the extended node in~$\mathcal T^*$ with
  embedding~$\mathcal E_\mu^*$.  As the only crossings
  in~$\mathcal E_\mu^*$ are between real edges, we can form its
  planarization by replacing each crossing by a dummy vertex, thereby
  subdividing the involved directed edges into paths of length~2.
  Since the upward embedding~$\mathcal E$ in particular induces an
  upward embedding of its planarization, it follows that the
  planarization of~$\mathcal E_\mu^\star$ does not contain a cycle
  that consists of directed edges, i.e.,~$\mathcal E_\mu$ is acyclic.
\end{proof}

\begin{lemma}[Sufficiency]
  \label{le:acyclic-skeletons-implies-upwardG}
  Let~$\{\mathcal E_\mu\}_{\mu \in V(\mathcal T)}$ be embeddings of
  the skeletons of~$\mathcal T$ that satisfy
  Conditions~\ref{c1}--\ref{c5}.  If~$\mathcal E_\mu$ is acyclic for all P- and R-nodes,
  then the combined embedding of
  $\{\mathcal E_\mu\}_{\mu \in V(\mathcal T)}$ is an outer-1-planar
  embedding of $G$ whose planarization is acyclic.
\end{lemma}

\begin{proof}
  Let~$\mathcal T^\star$ be the extension of~$\mathcal T$ with respect
  to the embeddings $\{\mathcal E_\mu\}_{\mu \in V(\mathcal T)}$ of
  the skeletons and let~$\mathcal E$ denote the combined embedding of
  $\{\mathcal E_\mu\}_{\mu \in V(\mathcal T)}$. By
  \Cref{thm:spqr-outer-1-planar}, $\mathcal E$ is an outer-1-planar
  embedding of $G$.  In the remainder, we show that its planarization
  $G'$ is acyclic. Thus, assume for the sake of contradiction that
  $G'$ contains a directed cycle~$C$.

  We will consider the \emph{projection} of~$C$ to the planarizations
  of the skeletons of the nodes of~$\mathcal T^\star$.
  Let~$\mu^\star$ be a node of~$\mathcal T^\star$ and consider its
  planarization~$\skel^+(\mu^\star)$.  An edge~$e$
  of~$\skel^+(\mu^\star)$ belongs to the projection of~$C$ if either
  (i) $e$ is a real edge that belongs to~$C$ or (ii) $e$ is a virtual
  edge whose expansion graph contains a real edge that belongs to~$C$.
  Observe that for each planarized skeleton, the projection of~$C$ is
  either a cycle~$C'$ or a single edge.  In the former case, we also
  say that~$C$ \emph{projects} to the cycle~$C'$.
    
  \textit{Claim:} There exists a P- or R-node~$\mu^\star$ of
  $\mathcal T^\star$ where the cycle~$C$ projects to a cycle~$C'$ in
  the planarization $\skel^+(\mu^\star)$ such that the source~$s$ of
  $G$ is either a vertex of~$\skel^+(\mu^\star)$ or belongs to the
  expansion graph of an edge of~$\skel^+(\mu^\star)$ that does not
  belong to~$C'$.
	
  To prove the claim, observe that the skeletons where~$C$ projects to
  a cycle form a subtree~$\mathcal T_C$ of~$\mathcal
  T^\star$. Similarly, the nodes whose skeletons contain~$s$ as a
  vertex also form a subtree~$\mathcal T_s$ of~$\mathcal T^\star$.  We
  choose~$\mu^\star$ either as a common node of these two subtrees (if
  one exists) or as the node of~$\mathcal T_C$ that is closest to the
  subtree~$\mathcal T_s$.  In the former case, we have that~$s$ is a
  vertex of~$\skel^+(\mu^\star)$.  Note that, since~$s$ is a source,
  it is not contained in~$C$ and therefore it is also not contained in
  the projection of~$C$ to~$\skel^+(\mu^\star)$.  In particular~$\mu^\star$ cannot be an S-node.  In the latter case
  let~$e$ denote the virtual edge of~$\skel^+(\mu^\star)$ whose
  expansion graph contains~$s$.  If~$e$ is contained in the projection
  of~$C$ to~$\skel^+(\mu^\star)$, then~$\refn(e)$ belongs to
  $\mathcal T_C$ and is closer to~$\mathcal T_s$ than~$\mu^\star$,
  which contradicts the choice of~$\mu^\star$.  Note that this in particular excludes the case that~$\mu^\star$ is an S-node.  This completes the
  proof of the claim.
	
  Consider the node~$\mu^\star$ from the claim and let~$\mu$ be the
  corresponding node of~$\mathcal T$.  
  The projection~$C'$ of~$C$
  to~$\skel^+(\mu^\star)$ may contain both real and virtual edges.  If
  all edges of~$C'$ are directed, the cycle is already present
  in~$\skel^+(\mu^\star)$, i.e., the embedding~$\mathcal E_\mu$
  of~$\skel(\mu)$ is not acyclic; a contradiction.  Since real edges
  are always directed, it hence follows that~$C'$ contains an
  undirected virtual edge~$\{u,v\}$ of $\skel(\mu^\star)$.  Since in
  expanded skeletons only real edges have crossings, the
  edge~$\{u,v\}$ is uncrossed in the embedding of~$\skel(\mu^\star)$.
  Therefore~$\expn(\{u,v\})$ is not crossed by edges outside
  of~$\expn(\{u,v\})$.  In particular the planarization
  of~$\expn(\{u,v\})$ is connected to the rest of the graph only via the vertices~$u$ and~$v$
  in the planarization $G'$ of $G$.  Since~$\{u,v\}$ belongs
  to~$C'$, the planarization of~$\expn(\{u,v\})$ contains a directed
  path~$\pi$ between its two poles, say from~$u$ to~$v$.  However,
  since~$\{u,v\}$ is undirected and does not contain $s$, both~$u$
  and~$v$ are sources in~$\expn(\{u,v\})$ and hence they are sources
  also in its planarization.  This contradicts the existence of~$\pi$,
  and hence proves the lemma.
\end{proof}

In light
of~\cref{le:upwardG-implies-acyclic-skeletons,le:acyclic-skeletons-implies-upwardG},
it suffices to test whether the skeletons of~$\mathcal T$ admit embeddings that satisfy~\ref{c1}--\ref{c5} and that are acyclic for P- and R-nodes.  We note
that~\ref{c1}--\ref{c4} and acyclicity are local conditions that can
be checked and all solutions can be enumerated\footnote{Recall that
  P-node skeletons have at most five edges and each R-node skeleton is
  a $K_4$.} locally and independently for each skeleton.  On the other
hand~\ref{c5} is a global property, which states that for each real
edge~$e$ that connects two P-nodes in a series, only one of them may
put a crossing on~$e$.

For each node~$\mu$ of~$\mathcal T$, let~$\mathcal F_\mu$ be a subset of the embeddings of~$\skel(\mu)$ that satisfy~\ref{c1}--\ref{c4}.  We call~$\mathcal F_\mu$ the \emph{feasible embeddings} of~$\skel(\mu)$.  We are interested in whether we can choose for each node~$\mu$ a feasible embedding~$\mathcal E_\mu \in \mathcal F_\mu$ such that together they satisfy~\ref{c5}.  We call such a choice of embeddings \emph{consistent}.

To find a consistent choice of embeddings, we construct an auxiliary graph~$H$. The vertex set of~$H$  is~$\bigcup_{\mu \in \mathcal T} \mathcal F_\mu$.  We turn each set~$\mathcal F_\mu$ into a clique and we connect two embedding choices of different P-nodes if choosing both of them simultaneously violates~\ref{c5}.  It is then clear that~$H$ contains an independent set whose size equals the number of nodes of $\mathcal T$ if and only if there exists a consistent choice of embeddings.  In the following, we show that this can be decided in linear time by proving that $H$ has bounded treewidth~\citep*{DBLP:books/sp/CyganFKLMPPS15}.

\begin{lemma}
	The auxiliary graph~$H$ has bounded treewidth.
\end{lemma}
\begin{proof}
    We first consider the following construction.
	Let $T$ be a tree, let~$I$ be an independent set of vertices in~$T$, and let~$c$ be a constant. 
	Let~$T'$ be the graph obtained by connecting for each vertex in~$I$ its neighbors by a cycle that visits them in some arbitrary order.  We say that~$T'$ is a \emph{closure} of~$T$ with respect to~$I$.   
	Since~$I$ is an independent set, each block of~$T'$ is either an edge or a wheel, and hence~$T'$ has treewidth at most~$3$.  Let~$T''$ be obtained from~$T'$ by expanding each vertex~$v$ of~$T'$ into a clique~$C_v$ of size at most~$c$ such that each edge~$uv$ of~$T'$ is expanded into a biclique completely connecting~$C_u$ and~$C_v$.  We call $T''$ the \emph{$c$-clique expansion} of~$T'$.  Clearly, the treewidth of~$T''$ is at most~$3c$.

    We now use this to bound the treewidth of $H$.
	For each S-node~$\lambda$ of~$\mathcal T$ the skeleton~$\skel(\lambda)$ defines a circular ordering of its virtual edges.  Note that embeddings of two different P-nodes~$\mu$ and~$\mu'$ can be connected by an edge only if they have a common S-node neighbor~$\lambda$ such that their corresponding virtual edges in~$\skel(\lambda)$ are consecutive in~$\skel(\lambda)$.  Therefore the auxiliary graph~$H$ is a subgraph of the graph obtained from the SPQR-tree~$\mathcal T$ by (i) forming the closure~$T'$ of~$\mathcal T$ with respect to the (independent set of) S-nodes, where we connect the neighbors of the S-node in the order in which the corresponding virtual edges appear in the skeleton of the S-node and (ii) taking the $c$-clique-expansion of~$T'$ where~$c$ is an upper bound on the number of embeddings satisfying~{C1}--{C4} for any skeleton.  \cite{DBLP:journals/algorithmica/AuerBBGHNR16} show that~$c \le 12$.  Hence the treewidth of~$H$ is at most~$36$.
\end{proof}

\begin{theorem}
  There is a linear-time algorithm for testing whether a given
  single-source graph admits an upward outer-1-planar embedding.
\end{theorem}

\begin{proof}By \Cref{lem:o1p-bico} we may assume that~$G$ is
  biconnected.  We first compute the SPQR-tree~$\mathcal T$ of~$G$ in
  linear time~\citep{DBLP:conf/gd/GutwengerM00}.  Next, we check that
  the skeleton of each R-node of~$\mathcal T$ is a $K_4$, and that the
  skeleton of each P-node of~$\mathcal T$ has at most four virtual
  edges.  If this fails, we can reject the instance as it does not
  have an outer-1-planar
  embedding~\citep{DBLP:journals/algorithmica/AuerBBGHNR16}.  Next, we
  compute for each node~$\mu$ of~$\mathcal T$ the set~$\mathcal F_\mu$
  that contains all acyclic embeddings of~$\skel(\mu)$ that satisfy
  conditions~\ref{c1}--\ref{c4}.  Since skeletons of S-nodes have a
  unique planar embedding (\ref{c2}) and the skeletons of P- and
  R-nodes have bounded size, this can be done in total linear time.
  It then remains to consistently choose these embeddings so that also
  Condition~\ref{c5} is satisfied.  To this end, we construct the
  auxiliary graph~$H$ and compute a maximum independent set, which
  takes linear time as well~\citep{DBLP:books/sp/CyganFKLMPPS15}.  If
  the size of the maximum independent set is smaller than the number
  of nodes of~$\mathcal T$, there is no consistent choice of
  embeddings and we reject the instance.  Otherwise, by \cref{thm:spqr-outer-1-planar}
  this choice defines an outer-1-planar embedding~$\mathcal E$.  By
  \cref{le:acyclic-skeletons-implies-upwardG} the
  embedding~$\mathcal E$ has an acyclic planarization and is hence
  upward outer-1-planar by \cref{th:characterization}.  We note that,
  in the positive case, the embedding~$\mathcal E$ can also be
  constructed in linear time.
\end{proof}

\section{Conclusion}
\label{se:conclusion}

In this paper we initiated the study of upward~$k$-planar drawings, that is, upward drawings of directed acyclic graphs such that every edge is crossed at most~$k$ times for a given constant~$k$. We first gave upper and lower bounds for the upward local crossing number of various graph families, i.e., the minimum~$k$ such that every graph from the respective family admits an  upward~$k$-planar drawing. 
We strengthen these combinatorial results by proving that testing a DAG for upward~$k$-planarity is NP-complete even for $k=1$. On the positive side, testing upward outer-1-planarity for  single source digraphs can be done in linear time.
We conclude the paper by listing some open problems that may stimulate further research.
\begin{enumerate}
	\item Is there a directed outerpath that does not admit an upward 1-planar drawing?  
	
	\item Consider the class $\mathcal{O}_\Delta$ of outerplanar graphs (or even 2-trees) of maximum degree $\Delta$.  Is there a function~$f$ such that every graph in $\mathcal{O}_\Delta$ admits an $f(\Delta)$-planar upward drawing?   
	
	\item In light of the lower bounds in \cref{sec:lower-bounds}, it is natural to consider graphs with a special
	structure, in order to prove sublinear upper bounds on their
	(upward) local crossing number. 
	For example, Wood and Telle~\cite[Corollary 8.3]{wt-pdcng-NYJM07} 
	show that every (undirected) graph of maximum
	degree~$\Delta$ and treewidth~$\tau$ admits a (straight-line) drawing
	in which every edge crosses
	$\mathcal O(\Delta^2\tau)$
	other edges.  Can the {\em upward} local crossing number be bounded similarly by a function in~$\Delta$ and~$\tau$?
	
	\item Do planar graphs of maximum degree $\Delta$ have upward local crossing number $\mathcal O(f(\Delta) n^{1-\epsilon})$ for some function $f$ and some constant $\epsilon > 0$?
	
	\item Can upward outer-1-planarity be efficiently tested for multi-source and multi-sink DAGs? 
	
	\item Investigate parameterized approaches to testing upward $1$-planarity. 
	
\end{enumerate}

\renewcommand{\emph}[1]{{\color{black}\em #1}}

\nocite{DBLP:conf/gd/AngeliniCBDFKS09,DBLP:conf/gd/BhoreLMN21,DBLP:conf/gd/BekosLFGMR22}

\pagebreak

\bibliographystyle{plainnat}
\bibliography{upwards}

\begin{thebibliography}{43}
\providecommand{\natexlab}[1]{#1}
\providecommand{\url}[1]{\texttt{#1}}
\expandafter\ifx\csname urlstyle\endcsname\relax
  \providecommand{\doi}[1]{doi: #1}\else
  \providecommand{\doi}{doi: \begingroup \urlstyle{rm}\Url}\fi

\bibitem[Angelini et~al.(2009)Angelini, Cittadini, {Di Battista}, Didimo,
  Frati, Kaufmann, and Symvonis]{DBLP:conf/gd/AngeliniCBDFKS09}
Patrizio Angelini, Luca Cittadini, Giuseppe {Di Battista}, Walter Didimo,
  Fabrizio Frati, Michael Kaufmann, and Antonios Symvonis.
\newblock On the perspectives opened by right angle crossing drawings.
\newblock In David Eppstein and Emden~R. Gansner, editors, \emph{Proc. 17th
  Int. Symp. Graph Drawing (GD)}, volume 5849 of \emph{LNCS}, pages 21--32.
  Springer, 2009.
\newblock \doi{10.1007/978-3-642-11805-0_5}.

\bibitem[Angelini et~al.(2011)Angelini, Cittadini, Didimo, Frati, {Di
  Battista}, Kaufmann, and Symvonis]{DBLP:journals/jgaa/AngeliniCDFBKS11}
Patrizio Angelini, Luca Cittadini, Walter Didimo, Fabrizio Frati, Giuseppe {Di
  Battista}, Michael Kaufmann, and Antonios Symvonis.
\newblock On the perspectives opened by right angle crossing drawings.
\newblock \emph{J. Graph Algorithms Appl.}, 15\penalty0 (1):\penalty0 53--78,
  2011.
\newblock \doi{10.7155/JGAA.00217}.

\bibitem[Auer et~al.(2016)Auer, Bachmaier, Brandenburg, Glei{\ss}ner, Hanauer,
  Neuwirth, and Reislhuber]{DBLP:journals/algorithmica/AuerBBGHNR16}
Christopher Auer, Christian Bachmaier, Franz~J. Brandenburg, Andreas
  Glei{\ss}ner, Kathrin Hanauer, Daniel Neuwirth, and Josef Reislhuber.
\newblock Outer 1-planar graphs.
\newblock \emph{Algorithmica}, 74\penalty0 (4):\penalty0 1293--1320, 2016.
\newblock \doi{10.1007/S00453-015-0002-1}.

\bibitem[Bekos et~al.(2022)Bekos, {Da Lozzo}, Frati, Gronemann, Mchedlidze, and
  Raftopoulou]{DBLP:conf/gd/BekosLFGMR22}
Michael~A. Bekos, Giordano {Da Lozzo}, Fabrizio Frati, Martin Gronemann, Tamara
  Mchedlidze, and Chrysanthi~N. Raftopoulou.
\newblock Recognizing {DAG}s with page-number 2 is {NP}-complete.
\newblock In Patrizio Angelini and Reinhard {von Hanxleden}, editors,
  \emph{Proc. 30th Int. Symp. Graph Drawing \& Network Vis (GD)}, volume 13764
  of \emph{LNCS}, pages 361--370. Springer, 2022.
\newblock \doi{10.1007/978-3-031-22203-0_26}.

\bibitem[Bekos et~al.(2023)Bekos, {Da Lozzo}, Frati, Gronemann, Mchedlidze, and
  Raftopoulou]{BekosLFGMR23}
Michael~A. Bekos, Giordano {Da Lozzo}, Fabrizio Frati, Martin Gronemann, Tamara
  Mchedlidze, and Chrysanthi~N. Raftopoulou.
\newblock Recognizing {DAG}s with page-number 2 is {NP}-complete.
\newblock \emph{Theor. Comput. Sci.}, 946:\penalty0 113689, 2023.
\newblock \doi{10.1016/J.TCS.2023.113689}.

\bibitem[Bertolazzi et~al.(1994)Bertolazzi, {Di Battista}, Liotta, and
  Mannino]{DBLP:journals/algorithmica/BertolazziBLM94}
Paola Bertolazzi, Giuseppe {Di Battista}, Giuseppe Liotta, and Carlo Mannino.
\newblock Upward drawings of triconnected digraphs.
\newblock \emph{Algorithmica}, 12\penalty0 (6):\penalty0 476--497, 1994.
\newblock \doi{10.1007/BF01188716}.

\bibitem[Bertolazzi et~al.(1998)Bertolazzi, {Di Battista}, Mannino, and
  Tamassia]{bertolazzi-siam}
Paola Bertolazzi, Giuseppe {Di Battista}, Carlo Mannino, and Roberto Tamassia.
\newblock Optimal upward planarity testing of single-source digraphs.
\newblock \emph{SIAM J. Comput.}, 27\penalty0 (1):\penalty0 132--169, 1998.
\newblock \doi{10.1137/S0097539794279626}.

\bibitem[Bhore et~al.(2021)Bhore, {Da Lozzo}, Montecchiani, and
  N{\"{o}}llenburg]{DBLP:conf/gd/BhoreLMN21}
Sujoy Bhore, Giordano {Da Lozzo}, Fabrizio Montecchiani, and Martin
  N{\"{o}}llenburg.
\newblock On the upward book thickness problem: Combinatorial and complexity
  results.
\newblock In Helen~C. Purchase and Ignaz Rutter, editors, \emph{Proc. 29th Int.
  Symp. Graph Drawing \& Network Vis. (GD)}, volume 12868 of \emph{LNCS}, pages
  242--256. Springer, 2021.
\newblock \doi{10.1007/978-3-030-92931-2_18}.

\bibitem[Bhore et~al.(2023)Bhore, {Da Lozzo}, Montecchiani, and
  N{\"{o}}llenburg]{BhoreLMN23}
Sujoy Bhore, Giordano {Da Lozzo}, Fabrizio Montecchiani, and Martin
  N{\"{o}}llenburg.
\newblock On the upward book thickness problem: Combinatorial and complexity
  results.
\newblock \emph{Eur. J. Comb.}, 110:\penalty0 103662, 2023.
\newblock \doi{10.1016/J.EJC.2022.103662}.

\bibitem[Binucci et~al.(2023)Binucci, {Da Lozzo}, {Di Giacomo}, Didimo,
  Mchedlidze, and Patrignani]{BinucciLGDMP19}
Carla Binucci, Giordano {Da Lozzo}, Emilio {Di Giacomo}, Walter Didimo, Tamara
  Mchedlidze, and Maurizio Patrignani.
\newblock Upward book embeddability of st-graphs: Complexity and algorithms.
\newblock \emph{Algorithmica}, 85\penalty0 (12):\penalty0 3521--3571, 2023.
\newblock \doi{10.1007/S00453-023-01142-Y}.

\bibitem[B{\"o}ttcher et~al.(2010)B{\"o}ttcher, Pruessmann, Taraz, and
  W{\"u}rfl]{BPTW2010}
Julia B{\"o}ttcher, Klaas~P. Pruessmann, Anusch Taraz, and Andreas W{\"u}rfl.
\newblock Bandwidth, expansion, treewidth, separators and universality for
  bounded-degree graphs.
\newblock \emph{Europ. J. Combin.}, 31\penalty0 (5):\penalty0 1217--1227, 2010.
\newblock \doi{10.1016/J.EJC.2009.10.010}.

\bibitem[Chaplick et~al.(2022{\natexlab{a}})Chaplick, {Di Giacomo}, Frati,
  Ganian, Raftopoulou, and Simonov]{ChaplickGFGRS22}
Steven Chaplick, Emilio {Di Giacomo}, Fabrizio Frati, Robert Ganian,
  Chrysanthi~N. Raftopoulou, and Kirill Simonov.
\newblock Parameterized algorithms for upward planarity.
\newblock In Xavier Goaoc and Michael Kerber, editors, \emph{Proc. 38th Int.
  Symp. Comput. Geom. (SoCG)}, volume 224 of \emph{LIPIcs}, pages 26:1--26:16.
  Schloss Dagstuhl~-- Leibniz-Zentrum f{\"{u}}r Informatik, 2022{\natexlab{a}}.
\newblock \doi{10.4230/LIPICS.SOCG.2022.26}.

\bibitem[Chaplick et~al.(2022{\natexlab{b}})Chaplick, {Di Giacomo}, Frati,
  Ganian, Raftopoulou, and Simonov]{cdf-tup-gd-2022}
Steven Chaplick, Emilio {Di Giacomo}, Fabrizio Frati, Robert Ganian,
  Chrysanthi~N. Raftopoulou, and Kirill Simonov.
\newblock Testing upward planarity of partial 2-trees.
\newblock In Patrizio Angelini and Reinhard von Hanxleden, editors, \emph{Proc.
  30th Int. Symp. Graph Drawing \& Netw. Vis. (GD)}, volume 13764 of
  \emph{LNCS}, pages 175--187, 2022{\natexlab{b}}.
\newblock \doi{10.1007/978-3-031-22203-0_13}.

\bibitem[Cygan et~al.(2015)Cygan, Fomin, Kowalik, Lokshtanov, Marx, Pilipczuk,
  Pilipczuk, and Saurabh]{DBLP:books/sp/CyganFKLMPPS15}
Marek Cygan, Fedor~V. Fomin, {\L}ukasz Kowalik, Daniel Lokshtanov, D{\'{a}}niel
  Marx, Marcin Pilipczuk, Michal Pilipczuk, and Saket Saurabh.
\newblock \emph{Parameterized Algorithms}.
\newblock Springer, 2015.
\newblock \doi{10.1007/978-3-319-21275-3}.

\bibitem[{Di Battista} and Tamassia(1988)]{batTam:88}
Giuseppe {Di Battista} and Roberto Tamassia.
\newblock Algorithms for plane representations of acyclic digraphs.
\newblock \emph{Theoret. Comput. Sci.}, 61:\penalty0 175--198, 1988.
\newblock \doi{10.1016/0304-3975(88)90123-5}.

\bibitem[{Di Battista} and Tamassia(1996)]{dt-olpt-96}
Giuseppe {Di Battista} and Roberto Tamassia.
\newblock On-line planarity testing.
\newblock \emph{SIAM J. Comput.}, 25:\penalty0 956--997, 1996.
\newblock \doi{10.1137/S0097539794280736}.

\bibitem[{Di Battista} et~al.(1999){Di Battista}, Eades, Tamassia, and
  Tollis]{dett-gd-99}
Guiseppe {Di Battista}, Peter Eades, Roberto Tamassia, and Ioannis~G. Tollis.
\newblock \emph{Graph Drawing}.
\newblock Prentice Hall, Upper Saddle River, NJ, 1999.

\bibitem[Didimo et~al.(2010)Didimo, Giordano, and Liotta]{upward-spirality}
Walter Didimo, Francesco Giordano, and Giuseppe Liotta.
\newblock Upward spirality and upward planarity testing.
\newblock \emph{SIAM J. Discrete Math.}, 23\penalty0 (4):\penalty0 1842--1899,
  2010.
\newblock \doi{10.1137/070696854}.

\bibitem[Didimo et~al.(2019)Didimo, Liotta, and
  Montecchiani]{DBLP:journals/csur/DidimoLM19}
Walter Didimo, Giuseppe Liotta, and Fabrizio Montecchiani.
\newblock A survey on graph drawing beyond planarity.
\newblock \emph{{ACM} Comput. Surv.}, 52\penalty0 (1):\penalty0 4:1--4:37,
  2019.
\newblock \doi{10.1145/3301281}.

\bibitem[Dujmovi\'c et~al.(2008)Dujmovi\'c, Kawarabayashi, Mohar, and
  Wood]{dujmovic_etal:scg08}
Vida Dujmovi\'c, Ken{-}ichi Kawarabayashi, Bojan Mohar, and David~R. Wood.
\newblock Improved upper bounds on the crossing number.
\newblock In Monique Teillaud, editor, \emph{Proc. 24th {ACM} Symp. Comput.
  Geom. (SoCG)}, pages 375--384, 2008.
\newblock \doi{10.1145/1377676.1377739}.

\bibitem[Fomin and Golovach(2003)]{FG2003}
Fedor~V. Fomin and Petr~A. Golovach.
\newblock Interval degree and bandwidth of a graph.
\newblock \emph{Discrete Appl. Math.}, 129\penalty0 (2--3):\penalty0 345--359,
  2003.
\newblock \doi{10.1016/S0166-218X(02)00574-7}.

\bibitem[Frati et~al.(2011)Frati, Fulek, and
  Ruiz{-}Vargas]{DBLP:conf/gd/FratiFR11}
Fabrizio Frati, Radoslav Fulek, and Andres~J. Ruiz{-}Vargas.
\newblock On the page number of upward planar directed acyclic graphs.
\newblock In Marc~J. {van Kreveld} and Bettina Speckmann, editors, \emph{Proc.
  19th Int. Symp. Graph Drawing (GD)}, volume 7034 of \emph{LNCS}, pages
  391--402. Springer, 2011.
\newblock \doi{10.1007/978-3-642-25878-7_37}.

\bibitem[Frati et~al.(2013)Frati, Fulek, and Ruiz{-}Vargas]{FratiFR13}
Fabrizio Frati, Radoslav Fulek, and Andres~J. Ruiz{-}Vargas.
\newblock On the page number of upward planar directed acyclic graphs.
\newblock \emph{J. Graph Algorithms Appl.}, 17\penalty0 (3):\penalty0 221--244,
  2013.
\newblock \doi{10.7155/JGAA.00292}.

\bibitem[Fulek et~al.(2013)Fulek, Pelsmajer, Schaefer, and
  {\v{S}}tefankovi{\v{c}}]{FPSS2013}
Radoslav Fulek, Michael~J. Pelsmajer, Marcus Schaefer, and Daniel
  {\v{S}}tefankovi{\v{c}}.
\newblock {Hanani--Tutte}, monotone drawings, and level-planarity.
\newblock In J{\'a}nos Pach, editor, \emph{Thirty Essays on Geometric Graph
  Theory}, pages 263--287. Springer, 2013.
\newblock \doi{10.1007/978-1-4614-0110-0_14}.

\bibitem[Garey and Johnson(1979)]{garey1979computers}
Michael~R. Garey and David~S. Johnson.
\newblock \emph{Computers and Intractability: A Guide to the Theory of
  NP-Completeness}.
\newblock W. H. Freeman, 1979.

\bibitem[Garg and Tamassia(1995)]{gt-upt-95}
Ashim Garg and Roberto Tamassia.
\newblock Upward planarity testing.
\newblock \emph{Order}, 12:\penalty0 109--133, 1995.
\newblock \doi{10.1007/BF01108622}.

\bibitem[Garg and Tamassia(2001)]{DBLP:journals/siamcomp/GargT01}
Ashim Garg and Roberto Tamassia.
\newblock On the computational complexity of upward and rectilinear planarity
  testing.
\newblock \emph{{SIAM} J. Comput.}, 31\penalty0 (2):\penalty0 601--625, 2001.
\newblock \doi{10.1137/S0097539794277123}.

\bibitem[Giordano et~al.(2015)Giordano, Liotta, Mchedlidze, Symvonis, and
  Whitesides]{GiordanoLMSW15}
Francesco Giordano, Giuseppe Liotta, Tamara Mchedlidze, Antonios Symvonis, and
  Sue Whitesides.
\newblock Computing upward topological book embeddings of upward planar
  digraphs.
\newblock \emph{J. Discrete Algorithms}, 30:\penalty0 45--69, 2015.
\newblock \doi{10.1016/J.JDA.2014.11.006}.

\bibitem[Gutwenger and Mutzel(2000)]{DBLP:conf/gd/GutwengerM00}
Carsten Gutwenger and Petra Mutzel.
\newblock A linear time implementation of {SPQR}-trees.
\newblock In Joe Marks, editor, \emph{Proc. 8th Int. Symp. Graph Drawing (GD)},
  volume 1984 of \emph{LNCS}, pages 77--90. Springer, 2000.
\newblock \doi{10.1007/3-540-44541-2_8}.

\bibitem[Heath et~al.(1999)Heath, Pemmaraju, and Trenk]{HeathPT99}
Lenwood~S. Heath, Sriram~V. Pemmaraju, and Ann~N. Trenk.
\newblock Stack and queue layouts of directed acyclic graphs: Part {I}.
\newblock \emph{{SIAM} J. Comput.}, 28\penalty0 (4):\penalty0 1510--1539, 1999.
\newblock \doi{10.1137/S0097539795280287}.

\bibitem[Hong and Tokuyama(2020)]{DBLP:books/sp/20/HT2020}
Seok{-}Hee Hong and Takeshi Tokuyama, editors.
\newblock \emph{Beyond Planar Graphs}.
\newblock Communications of {NII} Shonan Meetings. Springer, 2020.
\newblock \doi{10.1007/978-981-15-6533-5}.

\bibitem[Hong et~al.(2013)Hong, Eades, Katoh, Liotta, Schweitzer, and
  Suzuki]{DBLP:conf/gd/HongEKLSS13}
Seok{-}Hee Hong, Peter Eades, Naoki Katoh, Giuseppe Liotta, Pascal Schweitzer,
  and Yusuke Suzuki.
\newblock A linear-time algorithm for testing outer-1-planarity.
\newblock In Stephen~K. Wismath and Alexander Wolff, editors, \emph{Proc. 21st
  Int. Symp. Graph Drawing (GD)}, volume 8242 of \emph{LNCS}, pages 71--82.
  Springer, 2013.
\newblock \doi{10.1007/978-3-319-03841-4_7}.

\bibitem[Hong et~al.(2015)Hong, Eades, Katoh, Liotta, Schweitzer, and
  Suzuki]{DBLP:journals/algorithmica/HongEKLSS15}
Seok{-}Hee Hong, Peter Eades, Naoki Katoh, Giuseppe Liotta, Pascal Schweitzer,
  and Yusuke Suzuki.
\newblock A linear-time algorithm for testing outer-1-planarity.
\newblock \emph{Algorithmica}, 72\penalty0 (4):\penalty0 1033--1054, 2015.
\newblock \doi{10.1007/S00453-014-9890-8}.

\bibitem[Hopcroft and Tarjan(1973)]{DBLP:journals/siamcomp/HopcroftT73}
John~E. Hopcroft and Robert~Endre Tarjan.
\newblock Dividing a graph into triconnected components.
\newblock \emph{{SIAM} J. Comput.}, 2\penalty0 (3):\penalty0 135--158, 1973.
\newblock \doi{10.1137/0202012}.

\bibitem[Jungeblut et~al.(2023)Jungeblut, Merker, and Ueckerdt]{JungeblutMU23}
Paul Jungeblut, Laura Merker, and Torsten Ueckerdt.
\newblock Directed acyclic outerplanar graphs have constant stack number.
\newblock In \emph{Proc. 64th {IEEE} Ann. Symp. Foundat. Comput. Sci. (FOCS)},
  pages 1937--1952, 2023.
\newblock \doi{10.1109/FOCS57990.2023.00118}.

\bibitem[Kobourov et~al.(2017)Kobourov, Liotta, and
  Montecchiani]{DBLP:journals/csr/KobourovLM17}
Stephen~G. Kobourov, Giuseppe Liotta, and Fabrizio Montecchiani.
\newblock An annotated bibliography on 1-planarity.
\newblock \emph{Comput. Sci. Rev.}, 25:\penalty0 49--67, 2017.
\newblock \doi{10.1016/J.COSREV.2017.06.002}.

\bibitem[Mchedlidze and Symvonis(2009)]{MS2009}
Tamara Mchedlidze and Antonios Symvonis.
\newblock Crossing-optimal acyclic {Hamiltonian} path completion and its
  application to upward topological book embeddings.
\newblock In S.~Das and Ryuhei Uehara, editors, \emph{Proc. 3rd Workshop
  Algorithms \& Comput. (WALCOM)}, volume 5431 of \emph{LNCS}, pages 250--261.
  Springer, 2009.
\newblock \doi{10.1007/978-3-642-00202-1_22}.

\bibitem[N{\"{o}}llenburg and Pupyrev(2023)]{NollenburgP23}
Martin N{\"{o}}llenburg and Sergey Pupyrev.
\newblock On families of planar {DAGs} with constant stack number.
\newblock In Michael~A. Bekos and Markus Chimani, editors, \emph{Proc. 31st
  Int. Symp. Graph Drawing \& Netw. Vis. (GD)}, volume 14465 of \emph{LNCS},
  pages 135--151. Springer, 2023.
\newblock \doi{10.1007/978-3-031-49272-3_10}.

\bibitem[Papakostas(1994)]{p-uptod-GD94}
Achilleas Papakostas.
\newblock Upward planarity testing of outerplanar {DAGs}.
\newblock In Roberto Tamassia and Ioannis~G. Tollis, editors, \emph{Proc. Int.
  Sympos. Graph Drawing (GD)}, volume 894 of \emph{LNCS}, pages 298--306.
  Springer, 1994.
\newblock \doi{10.1007/3-540-58950-3_385}.

\bibitem[Schaefer(2024)]{schaefer:17-22}
Marcus Schaefer.
\newblock The graph crossing number and its variants: A survey.
\newblock \emph{Electr. J. Combin.}, Dynamic Survey DS21, 2024.
\newblock \doi{10.37236/2713}.

\bibitem[Valtr(2005)]{Valtr2007}
Pavel Valtr.
\newblock On the pair-crossing number.
\newblock In J.~E. Goodman, J.~Pach, and E.~Welzl, editors, \emph{Combinatorial
  and Computational Geometry}, volume~52 of \emph{MSRI Publications}, pages
  569--575. Cambridge University Press, 2005.
\newblock URL \url{https://library2.msri.org/books/Book52/files/31valtr.pdf}.

\bibitem[Wood(2006)]{Wood2006}
David~R. Wood.
\newblock Characterisations of intersection graphs by vertex orderings.
\newblock \emph{Australasian J. Combin.}, 34:\penalty0 261--268, 2006.
\newblock URL \url{https://ajc.maths.uq.edu.au/pdf/34/ajc_v34_p261.pdf}.

\bibitem[Wood and Telle(2007)]{wt-pdcng-NYJM07}
David~R. Wood and Jan~Arne Telle.
\newblock Planar decompositions and the crossing number of graphs with an
  excluded minor.
\newblock \emph{New York J. Math.}, 13:\penalty0 117--146, 2007.
\newblock URL \url{https://nyjm.albany.edu/j/2007/13-8.html}.

\end{thebibliography}

\end{document}